\documentclass[10pt]{article}
\usepackage[utf8]{inputenc}
\usepackage[english]{babel}
\usepackage[left=2.25cm,
            right=2.25cm,
            top=2.5cm,
            bottom=3cm
            ]{geometry}  
\usepackage[multiple]{footmisc}
\usepackage{amsthm}
\usepackage{amsmath}
\usepackage{amsfonts}
\usepackage{bbm} 
\usepackage{graphicx, caption, subcaption}
\usepackage[title]{appendix}
\usepackage{xcolor}
\usepackage[overload]{empheq}

\usepackage[backend=bibtex,sorting=none]{biblatex}

\DeclareMathOperator\erf{erf}

\newcommand{\mycomment}[1]{}
\newcommand{\lr}[1]{\left(#1\right)}
\newcommand{\slr}[1]{\left[#1\right]}
\newcommand{\expect}[1]{\mathbb E_{#1} }
\newcommand{\avg}[1]{\langle #1 \rangle}
\newcommand{\pder}[2][]{\ensuremath{\frac{\partial#1}{\partial#2}}} 
\newcommand{\der}[2][]{\ensuremath{\frac{d#1}{d#2}}} 
\newcommand{\bb}[1]{\boldsymbol{#1}}

\newtheorem{remark}{Remark}
\newtheorem{theorem}{Theorem}
\newtheorem{proposition}{Proposition}
\newtheorem{corollary}{Corollary}
\newtheorem{definition}{Definition}

\title{Guerra interpolation for place cells}
\author{Martino Salomone Centonze\footnote{Dipartimento di Matematica, Universit\`a di Bologna, Italy.}, Alessandro Treves\footnote{Cognitive Neuroscience, SISSA, Trieste, Italy.}, Elena Agliari\footnote{Dipartimento di Matematica, Sapienza Universit\`a di Roma, Italy}, Adriano Barra\footnote{Dipartimento di Scienze di Base ed Applicate per l'Ingegneria, Sapienza Universit\`a di Roma, Italy.}}

\date{October 2023}

\begin{document}

\maketitle

\begin{abstract}
Pyramidal cells that emit spikes when the animal is at specific locations of the environment are known as {\em place cells}: these neurons are thought to provide an internal representation of space via {\em cognitive maps}. Here, we consider the Battaglia-Treves neural network model for cognitive map storage and reconstruction, instantiated with McCulloch $\&$ Pitts binary neurons.
\newline
To quantify the information processing capabilities of these networks, we exploit spin-glass techniques based on Guerra's interpolation: in the low-storage regime (i.e., when the number of stored maps scales sub-linearly with the network size and the order parameters self-average around their means) we obtain an exact phase diagram in the  noise vs inhibition strength plane (in agreement with previous findings) by adapting the Hamilton-Jacobi PDE-approach. Conversely, in the high-storage regime, we find that -for mild inhibition and not too high noise- memorization and retrieval of an extensive number of spatial maps is indeed possible, since the maximal storage capacity is shown to be strictly positive. These results, holding under the replica-symmetry assumption, are obtained by adapting the standard interpolation based on stochastic stability and are further corroborated by Monte Carlo simulations (and replica-trick outcomes for the sake of completeness).  
\newline 
Finally, by relying upon an interpretation in terms of hidden units, in the last part of the work, we adapt the Battaglia-Treves model to cope with more general frameworks, such as bats flying in long tunnels. 

\end{abstract}

\section{Introduction}\label{sec:intro}
The hippocampus contributes in the brain to spatial recognition and particularly spatial memory. In particular, in rodents, hippocampal cells recorded during free foraging and other spatial tasks have turned out to be mainly correlated with the animal's position; this was first observed for pyramidal neurons in the CA1 area a long time ago \cite{Keefe}: each neuron fires when the animal is in one or more regions of the current environment. These neurons were therefore called \emph{place cells}, and the regions of space in which they get active were called \emph{place fields}. In other words, place cells provide an internal representation of the animal position and additional evidence shows that these neural representations play a crucial role in spatial memory \cite{PC-review1,PC-review2}.
 
Continuous-attractor neural networks (CANNs) (see e.g., \cite{Coolen}) offer a natural tool for simulating such complex systems \cite{PlaceCellOriginal,MonassonPlaceCellsPRL,MonassonPlaceCellsLong}.
In general, in a recurrent attractor neural network, some information (e.g., different locations in a certain environment) are encoded in the firing patterns of neurons and, for a suitable setting of interaction strengths among neurons, these patterns correspond to stationary states (attractors) of network dynamics. Thus, a stimulus which elicits the retrieval of previously-stored information (e.g., a detail of an experienced location) is expected to yield a stationary (or approximatively stationary) pattern for neuronal activity that codifies for further information (e.g., the stored location in relation to others) allowing for its use in behaviour and possibly for its consolidation in long-term storage. However, unlike simple attractor models, such as the Hopfield network (which can operate with very distributed patterns of activity in which e.g. half of the binary units are active), in a CANN the interaction between two units necessarily includes along with an excitatory term -- that depends on the similarity between their preferred stimuli (e.g., the proximity of their place fields) -- also a (long range) inhibition term -- that prevents all cells being active together. This arrangement enables a CANN to hold a continuous family of stationary states, rather than isolated ones. 
In the case of place cells, stationary states occur in the form of localized bump of activity (also referred to as \emph{coherent states}), peaked at a certain retrieved location. Ideally, retrievable locations span a continuous set of nearby values, although in practice finite size effects are known to impose a discretization or a roughening of the theoretically continuous manifold of attractor states \cite{Tsodyks1995,MonassonTreves2}.
This way a CANN is able to update its states (internal representations of stimuli) smoothly under the drive of an external input: several mathematical formulations for the generation of place-cells have been introduced in the past, as well as others describing the CANN that could be realized with populations of place cells \cite{Tsodyks1995,PlaceCellOriginal,battaglia1998attractor}. Much work has focused on the deviations expected in real networks from the idealized models \cite{Tsodyks1995,MonassonTreves2}.  Despite the intensive investigations that have been (and are still being\footnote{See e.g. \cite{Ale1,AABO-JPA2020,Contucci1,Contucci2,Jean1,Jean2,AuroBea1,AuroBea2,Fede1,Fede2,Alessia1,Alessia2,Dani1,Dani2,CarloMarc,Krotov,Kanter,Lenka,Huang,Taro1,Taro2} for a glance on recent findings (with a weak bias toward rigorous approaches).}) carried on by the statistical-mechanics community on Hebbian architectures for neural networks \cite{Amit,Rolls,Coolen},  no rigorous results have been obtained dealing with place cells and one of the goals of the current work is to fill this gap by rigorously analysing a model for a CANN with hippocampal place cells. 
Specifically, we focus on the model introduced by Battaglia and Treves \cite{battaglia1998attractor} but here, instead of considering real-valued activities for model neurons as in the original work,  units are binary variables, taking value $0$ (when resting) or $1$ (when firing), namely McCulloch-Pitts neurons.
\newline
Having in mind a rodent freely moving within a given (limited) environment (e.g., a rat in a box), the $i$-th place cell is active, $s_i = +1$, when the animal is at a location of such an  environment within the place field of that unit, while, away from that location, the place cell stays silent, $s_i =0$.
\newline
Unlike the analysis in \cite{MonassonTreves2}, each unit is taken to have one and only one place field, of standard size $a$ (as a fraction of the size of the environment). Now, the closer two adjacent place fields centers, the stronger the correlation between the corresponding place cell activities. This means that, in modelling the related neural dynamics, excitatory interactions among place cells have to be a function of the proximity between their fields, while inhibition is assumed to operate globally, to ensure the low firing rates suitable for the required coding level, parametrized by the place field size. 
Moreover, multiple maps are assumed to be stored in the network (and this can be obtained by combining different place fields in a Hebbian-like fashion) so that the same neuron can play an active role in the reconstruction of an extensive number of these stored maps.
\newline
We investigate the emerging computational capabilities collectively shown by such a network, by adapting mathematical techniques based on Guerra's interpolation (originally developed for spin glasses \cite{guerra_broken}) to these CANNs. The ultimate goal is to prove that such a model shows distinct phases, in one of which the neural network, for a relatively wide range of inhibition levels, is able to retrieve coherent states that provide an abstract representation of the outside space. 
In particular, we derive results for both the {\em low storage} (in which the number of maps stored in the network is sub-linear in the network size $N$) and the {\em high storage} (where a number of maps of order $N$ is stored) regimes. The former is tackled by the {\em mechanical analogy} \cite{HJ-Barra2010Aldo,HJ-Barra2013Gino,Barra-JSP2008}, namely an interpolation that maps the free energy of the model into an effective action to be addressed with tools pertaining to analytical, rather than statistical, mechanics (i.e. the Hamilton-Jacobi PDE).  The latter is investigated via the one-body interpolation, that is the classical Guerra's approach based on stochastic stability \cite{GuerraNN,Longo,glassy}, under the classical assumption of replica symmetry (namely under the assumption that the order parameters of the theory self-average around their means in the large network-size limit), and, in this regime, we also estimate the maximum storage capacity. Note that, for completeness, in the Appendix we also provide a replica-trick calculation \cite{Amit,MPV}  and we obtain the same results derived in the main text via Guerra's techniques.

Finally, we also inspect the network's capabilities in different scenarios than those which originally inspired the model, following recent experimental findings.
In fact, while in the past experimental protocols involved rather confined environments (see e.g. \cite{Leutgeb,ThompsonRat}), recently recording from e.g. rats running on long tracks (see e.g. \cite{Fenton2008,Rich2014}) and from  bats flying in long tunnels (see e.g. \cite{Batman1,MonassonTreves2}) show cells with multiple place fields of highly variable width and peak rate. Both parameters appear to approximately follow a log-normal distribution. In the last part of this work we show computationally that a minimal adaptation of the Battaglia-Treves model captures much of this new scenario too. In particular, driven by mathematical modelling rather than neuro-scientific evidence, we explore a duality of representation of the Battaglia-Treves model in terms of a bipartite network (known as {\em restricted Boltzmann machine} in Machine Learning jargon \cite{Hinton1985}) equipped with hidden {\em chart cells} (one per stored chart) and we show how -while the animal moves in this long environment- the different chart cells of this dual representation are active sequentially one after the other, tiling the  environment the animal is exploring.
\newline
\newline
The work is structured as follows. In Sec.~\ref{sec:model} we introduce the model and the major findings that we obtained in the present study. In particular, in Sec.~\ref{Sec:2.1} we discuss the Battaglia-Treves model on the circular manifold we will study, in Sec.~\ref{Sec:2.2} we introduce the observables we need to build a theory for this model (i.e. order and control parameters), in Sec.~\ref{Sec:2.3} we show the output of the theory, summarized via phase diagrams in the space of the control parameters and, finally, in Sec.~\ref{Citua} we prove the equivalence of such a model to a restricted Boltzmann machine equipped with  single-map selectively firing hidden neurons.
\newline
Sec.~\ref{sec:trois} is entirely dedicated to the underlying methodological aspects of such investigation, treating separately the low-load (Sec.~\ref{LowStorage}) and the high-load (Sec.~\ref{Guerra}) regimes. Next, in Sec.~\ref{sec:Tunnel}, we generalize the model and show its robustness w.r.t. the way the maps are coded and, thus, stored. Conclusions and outlooks are presented in Sec.~\ref{Conclusions}. Finally, App.~\ref{Replicas} contains the replica-trick analysis for the sake of completeness while App.~\ref{GuerraAPP} contains technical details of the analytical results presented in the main text. 


\section{The model: from definitions to computational capabilities}\label{sec:model}

The Battaglia-Treves model was introduced to describe the behavior of pyramidal neurons in a rodent hippocampus. In the original model \cite{battaglia1998attractor,PlaceCellOriginal} neuronal activity is represented by  continuous variables  linearly activated by their input and thresholded at zero, while here we consider $N$ McCulloch-Pitts neurons, whose activity is denoted as $s_i \in  \{0,1\}$ for $i \in (1,...,N)$, in such a way that, when $s_i=+1$ ($s_i=0$), the $i$-th neuron is spiking (quiescent).
As anticipated in Sec.~\ref{sec:intro}, the model is designed in such a way that a neuron fires when the rodent is in a certain region of the environment. 
\newline
The latter is described by a manifold equipped with a metric and uniformly partitioned into $N$ regions of (roughly) equal size, centered in $\vec{r}_i$, for $i \in (1,...,N)$, which are the cores of their place fields. When the agent modelling the rodent happens to be in the $i$-th core place field, the corresponding $i$-th neuron is expected to fire along with other neurons with their core place fields close by. Note that the behavior of different neurons is not independent: in fact, when the distance $|\vec r_i- \vec r_j|$ between $\vec{r}_i$ and $\vec{r}_j$ is small enough, the reciprocal influence between neurons $i$ and $j$ is strong, and they are  likely to be simultaneously active.
This is captured by the definition of the interaction matrix (i.e., the synaptic coupling in a neural network jargon) $\bb J$, whose element $J_{ij}$ represents the interaction strength between the neurons $(i,j)$ and is taken proportional to a given kernel function $\mathcal K(|\vec r_i- \vec r_j|)$ depending on the distance between the related place fields:
\begin{equation} \label{eq:single}
    J_{ij} \propto \mathcal K (|\vec r_i- \vec r_j|).
\end{equation}
Notice that the core place fields in any given environment are assumed to have been already learnt or assigned, thus their coordinates do not vary with time (i.e., they are {\em quenched} variables in a statistical mechanics jargon), while the state $\bb s = (s_1, s_2, ..., s_N)$ of the neurons (i.e. their neural activity) is a dynamical quantity. 

Following \cite{battaglia1998attractor}, and as standard for Hebbian neural networks \cite{Hopfield,amit1985storing,AABO-JPA2020,bovier2012mathematical}, we consider multiple arrangements of place fields, also referred to as {\em charts} or {\em maps} corresponding to distinct environments, e.g., different experimental room in which a rodent has been left to forage \cite{TrevesPNAS}, such that each neuron participates in each chart, and thus can contribute to the recognition of several charts. The synaptic matrix, accounting for $K$ charts, is obtained by summing up many terms like the one in \eqref{eq:single}, that is   
\begin{equation}\label{eq:J}
    J_{ij} \propto \sum_{\mu=1}^K \mathcal K (|\vec r^\mu_i- \vec r^\mu_j|),
\end{equation}
where $\vec{r}_i^{\mu}$ represents the core of the $i$-th place field in the $\mu$-th map.
The kernel function $\mathcal K$ is taken the same for each chart (i.e., $\mathcal K_{\mu} \equiv \mathcal K$) and the variation from one chart to another is only given by a different arrangement of core place fields  $\{\vec r^\mu_i\}_{i=1,...,N}^{\mu=1,...,K}$, in practice a reshuffling. Remarkably, here, different {\em charts} play a role similar to that played by different \emph{patterns} in the standard Hopfield model \cite{Hopfield,AGS,amit1985storing} and, along the same lines, we assume that the stored maps are statistically independent \cite{TrevesPNAS}: in other words, the position of the core place field associated to a certain neuron in one map, say $\vec r_i^{\mu}$, is uncorrelated to the core place field associated to the same neuron in any other map, say $\vec r_i^{\nu}$, at least when the network storage load $K/N$ is small enough.\\
\begin{figure}[tb] 
    \centering
    \includegraphics[width=10cm]{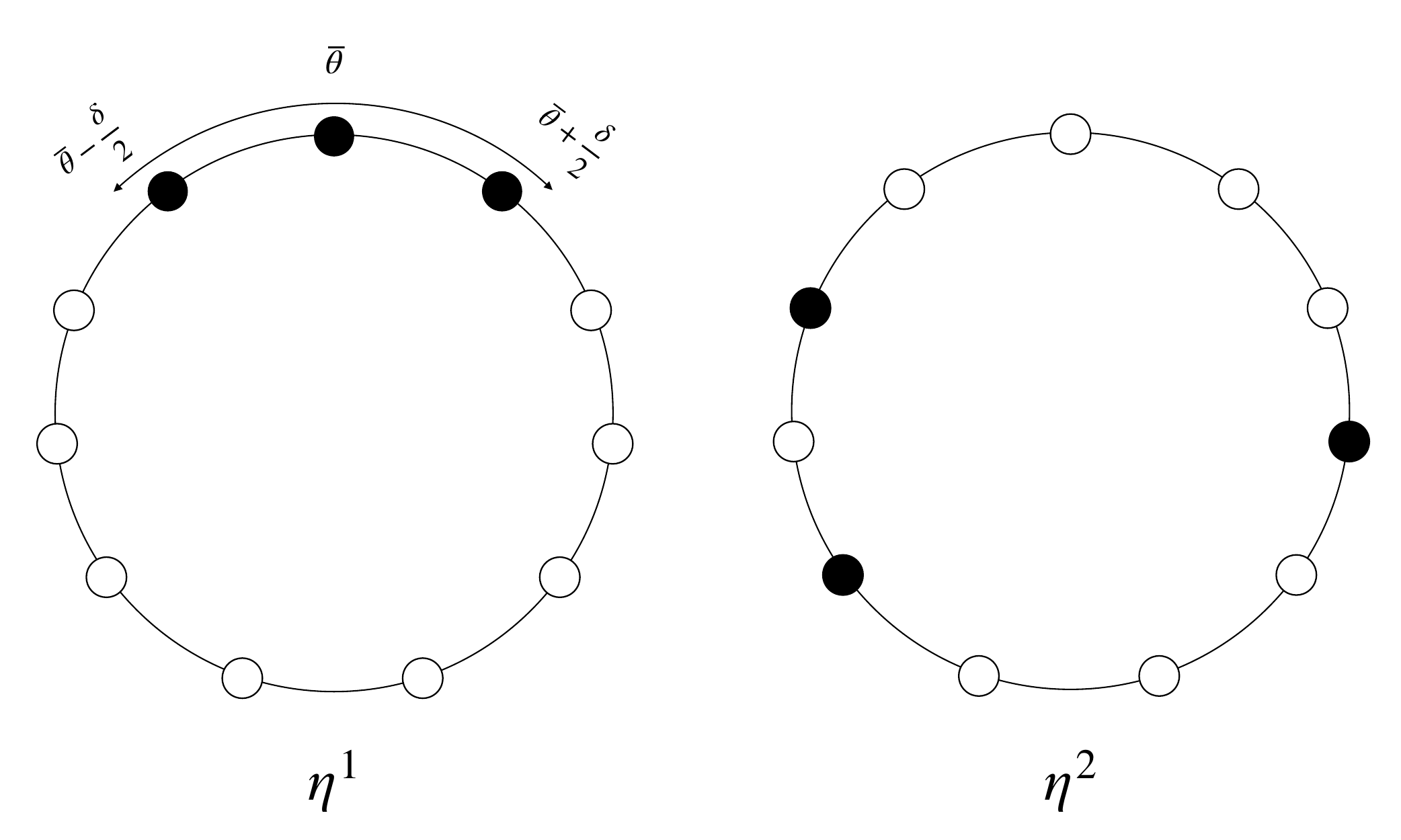}
    \caption{Two examples of maps $\eta^1$ and $\eta^2$. A coherent state  is shown in $\eta^1$ as all neurons that lie within the $[\overline \theta - \delta/2,\overline \theta + \delta/2]$ interval are activated (here displayed as black dots), while the others are quiescent (in white dots). The same firing pattern of neurons, that looks coherent in the first map $\eta^1$, looks disordered in the other map $\eta^2$. Note that the centers of the place fields are scattered roughly uniformly along the unitary circle $S^1$ and that the width of all the place fields is roughy the same in this first scenario.}
    \label{fig:placemaps}  
\end{figure}

\subsection{Coherent states, cost function, and other basic definitions}\label{Sec:2.1}
Let us now describe more explicitly our framework: for illustrative purposes, here we restrict ourselves to a one-dimensional manifold given by the unit circle $S^1$ (as often done previously, see e.g., \cite{MonassonPlaceCellsPRL,MonassonPlaceCellsLong,MonassonTreves19,MonassonTreves2}), in such a way that the position of the $i$-th place field in the $\mu$-th chart is specified by the angle $\theta^\mu_i\in [-\pi,\pi]$ w.r.t. a reference axes and it can thus be expressed by the unit vector $\vec \eta^\mu_i$ defined as
\begin{align}\label{eq:etas}
    \eta^\mu_i = \lr{\cos \theta_i^\mu, \sin \theta_i^\mu},
\end{align}
where we dropped the arrow on top to lighten the notation.  
\newline
In this coding, the trivial states $\bb s=(1,..,1)$ and $\bb s=(0,..,0)$, beyond being biologically unrealistic, do not carry information as they lead to a degenerate activation in such a way that the physical location of the animal cannot be represented by the network.
\newline
As opposed to these trivial states, we define
\begin{definition}\label{CoherentState}(Coherent state) 
Given a set of $N$ McCulloch-Pitts neurons constituting a CANN, their state $\boldsymbol s = (s_1, ..., s_N) \in \{0,1\}^N$ is said to be a coherent state centered around $\overline{\theta}\in[0,2\pi]$ with width $\delta$ if, for each $i=1,..,N$:
    \begin{align}
        s_i = \begin{cases}
            1,\:\:\: |\theta^1_i - \overline \theta|\leq \delta/2,\\
            0,\:\:\: |\theta^1_i - \overline \theta|>\delta/2,
        \end{cases}
    \end{align}
where $\theta^1_i$ is the coordinate of the $i$-th place cell in the first map $\mathbf \eta^1$ that we used as an example. In other words, a coherent state in a CANN plays the same role of a retrieval state in an ANN \cite{AGS}.
\end{definition}
Notice that the network is able to store several place fields for multiple maps and to retrieve a single place field, by relaxing on the related coherent state of neural activities: as shown in Figs.~\ref{fig:placemaps}-\ref{fig:bump}, such a coherent state for $\mu=1$ looks random in others maps (e.g. $\mu=2$ in the picture), provided that these are independent and that their number in the memory is not too large.  
Therefore, the application of an external stimulus in the $\mu-$th map corresponds to inputting the Cauchy condition for neural dynamics ({\em vide infra}) and by setting the initial neuronal configuration such that $s_i=+1$ if and only if the associated coordinate in the $\mu-$th map is in the interval $\theta_i^\mu \in [\overline{\theta}-\delta/2, \overline{\theta}+\delta/2]$, as sketched in Fig.~\ref{fig:placemaps}. 

Following \cite{battaglia1998attractor}, keeping in mind that the kernel has to be a function of a distance among place field cores on the manifold and that the latter is the unitary circle, we now make a specific choice for the interaction kernel that will be implemented in the network:  to take advantage from the {\em Hebbian experience} \cite{Amit}, the interacting strength between neurons will be written as\footnote{We can assume that the place field cores cover roughly uniformly the embedding space, namely the angles $\bar{\theta}^{\mu}$ are uniformly distributed along unitary circle (sampled from $\mathcal{U}_{[-\pi,\pi]}$), yet, there is no need to introduce a prior for them as, given the rotational invariance of the CANN kernel, a uniform dislocation of place fields is automatically fulfilled.}
\begin{align}\label{eq:kernel1}
    J_{ij} =  \frac{1}{N}\sum_{\mu=1}^K (\eta^\mu_i \cdot \eta^\mu_j).
\end{align}
\begin{figure}[!t] 
    \centering
    \includegraphics[width=16cm]{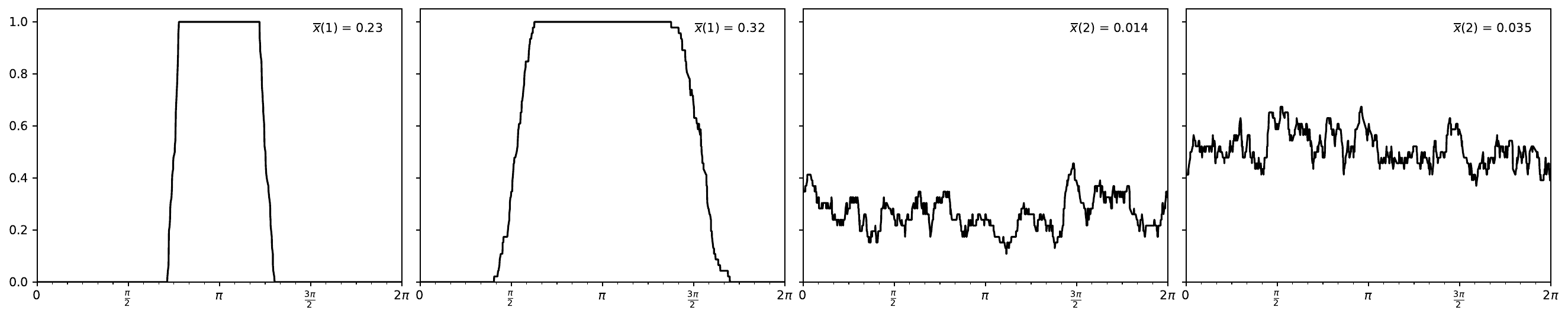}
    \caption{Neuronal activity in the first two maps $\mu=1,2$.
    The neurons coordinates $\theta_i^\mu\in[0,2\pi]$ in the given map are displayed along the x-axis, while the neuronal activity is shown on the y-axis  and computed as the spatial average of the neuronal states in a spatial window of fixed length).\\
    First panel from the left: we provide the model with the initial state $\mathbf s_0$ (a {\em bump} in the map $\mu=1$ centered around $\theta=\pi$): note that in other maps (e.g. $\mu=2$, third panel) this input appears as a random state. This Cauchy condition for the neural dynamics (i.e. the stochastic process \eqref{dinamical} driven by the Hamiltonian \eqref{eq:H_Hop}) allows the network to evolve toward a stationary state $\mathbf s_{out}$, that, in the map $\mu=1$, is the coherent state shown on the second panel from the \emph{left}, while it still appears random in other maps (as e.g. $\mu=2$ in the fourth panel). The map overlap $\overline x$ is also shown in the legend of each panel: coherent states have higher overlap than random states, as expected.}
    \label{fig:bump}  
\end{figure}
\\
Notice that the Hebbian kernel \eqref{eq:kernel1} is a function of the relative euclidean distance of the $i,j$ neuron's coordinates $\theta_i,\theta_j$ in each map $\mu$: to show this, one can simply compute the dot product as $\eta^\mu_i\cdot\eta^\mu_j=\cos(\theta_i^\mu)\cos(\theta_j^\mu)+\sin(\theta_i^\mu)\sin(\theta_j^\mu)=\cos(\theta_i^\mu-\theta_j^\mu)$.
\\
We can now define the  Cost function (or Hamiltonian, or energy in a physics jargon) of the model as given by the next
\begin{definition} \label{def:cost_function} (Cost function)
Given $N$ binary neurons $\bb s = (s_1, ..., s_N) \in \{ 0,1 \}^N$, $K$ charts $\bb \eta = (\eta^1,...,\eta^K)$ with  $\eta^{\mu} \in [-1, +1]^N$ for $\mu \in (1,...,K)$ encoded with the specific kernel \eqref{eq:kernel1}, and a free parameter $\lambda \in \mathbb{R}^+$ to tune the global inhibition within the network, the Hamiltonian for chart reconstruction reads as\footnote{The symbol `$\approx$' in eq. \eqref{eq:H_Hop} becomes an exact equality in the thermodynamic limit, $N\to \infty$, where, splitting the summation as $\sum_{i<j} = 1/2\sum_{i,j}^N + \sum_{i=1}^N$, the last term, being sub-linear in $N$, can be neglected.}
\begin{align}\label{eq:H_Hop}
    H_N(\boldsymbol s | \boldsymbol \eta) = - \sum_{i<j}^{N,N}J_{ij} s_i s_j + \frac{(\lambda-1)}{N}  \sum_{i<j}^{N,N} s_i s_j \approx -\frac{1}{2N} \sum_{\mu=1}^K \sum_{i,j=1}^{N,N} (\eta^\mu_i \cdot \eta^\mu_j) s_i s_j + \frac{(\lambda -1)}{2N}  \sum_{i,j=1}^{N,N} s_i s_j.
\end{align}
\end{definition}
Notice that the factor $N^{-1}$ in front of the sums ensures that the Hamiltonian is extensive in the thermodynamic limit $N\to \infty$ and the factor $1/2$ is inserted in order to count only once the contribution of each couple. 
Also, the hyper-parameter $\lambda$ tunes a source of inhibition acting homogeneously among all pairs of neurons and prevents the network from collapsing onto a fully firing state $\bb s = (1, ..., 1)$. In fact, for $\lambda \gg 1$ the last term at the r.h.s. of eq.~ \eqref{eq:H_Hop} prevails,  global inhibition dominates over local excitation and the most energetically-favorable configuration is the fully inhibited one (where all the neurons are quiescent $\bb s=(0,..,0)$); in the opposite limit, for $\lambda \ll 1$, the most energetically-favorable configuration is  the totally excitatory one (where all the neurons are firing $\bb s = (1, ...,1)$). 

\bigskip

Once a Hamiltonian is provided, it is possible to construct a Markov process for the neural dynamics  by introducing a source of noise $\beta \in \mathbb{R}^{+}$ (i.e., the temperature in a physics jargon) in the following master equation
\begin{equation}\label{dinamical}
\mathbf{P}_{t+1}(\boldsymbol{s}\vert \boldsymbol \eta) = \sum_{\boldsymbol{s'}} W_{\beta}(\boldsymbol{s},\boldsymbol{s'}\vert \boldsymbol \eta)\mathbf{P}_t(\boldsymbol{s'}\vert \boldsymbol \eta),
\end{equation}
where $\mathbf{P}_t(\boldsymbol{s} \vert \boldsymbol \eta)$ represents the probability of finding the system in a configuration of neural activities $\bb s$ at time $t$, being $\bb \eta$ the set of $K$ charts.  $W_{\beta}(\boldsymbol{s},\boldsymbol{s'}\vert \boldsymbol \eta)$ represents the transition rate, from a state $\boldsymbol{s'}$ to a state $\boldsymbol{s}$ and it is chosen in such a way that the system is likely to lower the energy \eqref{eq:H_Hop} along its evolution (see e.g. \cite{Coolen} for details): this likelihood is tuned by the parameter $\beta$ such that for $\beta \to 0^+$ the dynamics is a pure  random walk in the neural configuration space (and any configuration is equally likely to occur), while for $\beta \to +\infty$ the dynamics becomes a  
steepest descend toward the minima of the Hamiltonian and the latter (in this deterministic limit) plays as the Lyapounov function for the dynamical process \cite{Coolen}. Remarkably, the symmetry of the interaction matrix, i.e. $J_{ij}=J_{ji}$, is enough for detailed balance to hold 
such that the long-time limit of the stochastic process \eqref{dinamical} relaxes to the following Boltzmann-Gibbs distribution
\begin{equation}\label{BGmeasure}
\lim_{t \to \infty}\mathbf{P}_t(\boldsymbol{s} \vert \boldsymbol \eta) = \mathbf{P}(\boldsymbol{s} \vert \boldsymbol \eta) = \frac{e^{-\beta H_N(\boldsymbol s | \boldsymbol \eta)}}{Z_N(\beta, \boldsymbol \eta)}
\end{equation}
where $Z_N(\beta, \boldsymbol \eta)$ is the normalization factor, also referred to as {\em partition function}, as stated by the next 
\begin{definition}(Partition function)
Given the Hamiltonian $H_N(\boldsymbol s | \boldsymbol \eta)$ of the model \eqref{eq:H_Hop} and the control parameters $\beta$ and $\lambda$ ruling the neuronal dynamics, the associated partition function $Z_{N}(\beta, \lambda, \bb \eta)$ reads as
\begin{equation}
    Z_{N}(\beta, \lambda, \bb \eta) = \sum_{\{\bb s \}}^{2^N} e^{-\beta H(\bb s| \bb \eta)}.
\end{equation}
\end{definition}

\subsection{Order parameters, control parameters and other thermodynamical observables}\label{Sec:2.2}
We now proceed by introducing the macroscopic observables useful to describe the emerging behavior of the system under investigation.
\begin{definition}(Order parameters)
In order to assess the quality of the retrieval of a given chart $\boldsymbol \eta^{\mu}$, it is useful to introduce the two-dimensional \emph{order parameter} $x_\mu,\ \mu \in (1,...,K)$  defined as 
\begin{align}
\label{eq:one_order}
    x_\mu = \frac{1}{N} \sum_{i=1}^N \eta^\mu_i s_i \in [-1/\pi, +1/\pi]^2.
\end{align}
Another order parameter to consider, that does not depend on the place fields, is the mean activity of the whole population of neurons, defined as
\begin{align}
\label{eq:two_order}
    m = \frac{1}{N} \sum_{i=1}^N s_i \in [0, 1].
\end{align}
Finally, the last order parameter to introduce is the two-replica overlap  (whose usage is required solely when investigating the high-storage regime), defined as
\begin{align}
\label{eq:three_order}
    q_{ab}= \frac{1}{N} \sum_i s_i^a s_i^b \in [0, 1],
\end{align}
where $a$ and $b$ are the replica indices, namely they label different copies of the system, characterized by the same realization of maps (i.e., by the same quenched disorder in a spin-glass jargon). 
\end{definition}
Notice that, since we are working with Boolean neurons ($\bb s \in \{0,1\}^N$), $q_{11}$ precisely equals $m$ (the level of activity of the network) and the overlap $q_{12}$ can only take values in $0\leq m^2 \leq q_{12}\leq m\leq 1$. In particular, the limiting case $q_{12}=m$ denotes a frozen configuration (which can be a coherent state or not) while $q_{12}=m^2$ indicates that the two replicas are maximally uncorrelated. 
\newline
As a consistency check, we notice that the module of $x_\mu$, that is $|x_\mu| = \sqrt{x_\mu \cdot x_\mu}$, is strictly positive for the coherent states, while for the trivial states (i.e., the pure excitatory and the pure inhibitory configurations), it is vanishing.
%
When all the neurons are quiescent this holds straightforwardly, while, in the case where all neurons are firing, $x_\mu$ is just the sum of each coordinate vector $\eta^\mu_i$ in the given map and this quantity is close to zero (and exactly zero in the thermodynamic limit $N\to \infty$) when the coordinate vectors associated to each neuron in the given map are statistically independent and the angles $\theta^\mu_i$ sampled from a uniform distribution in the interval $[-\pi,\pi]$.
Conversely, for a coherent state, the vector $x_\mu$ is just the sum of the coordinates $\eta^\mu_i$ on the sites where the neurons are active, i.e., $x_\mu \equiv (x_\mu^{(1)}, x_\mu^{(2)}) = \frac{1}{N}\sum_{i|s_i = 1} (\cos \theta^\mu_i, \sin \theta^\mu_i) $. In the thermodynamic limit we can rewrite it as an integral, namely
\begin{align*}
    &x^{(1)}_\mu \underset{N \to \infty}{\sim} \frac{1}{2\pi} \int_{-\delta/2}^{\delta/2} d\theta \cos \theta = \frac{1}{\pi} \sin \delta/2,\\
    &x^{(2)}_\mu \underset{N \to \infty}{\sim} \frac{1}{2\pi} \int_{-\delta/2}^{\delta/2} d\theta \sin \theta = 0,
\end{align*}
hence we get $| x_\mu | = \frac{1}{\pi} \sin \delta/2$. The latter is positive for $0<\delta<2\pi$ and reaches its maximum at $\delta = \pi$, corresponding to $|x_\mu | \underset{\delta=\pi}{=} \frac{1}{\pi}\sim 0.32$. This situation represents a coherent state of size $\delta=\pi$. This simple result shows that $| x_\mu |$ is an informative order parameter for this network, as it reports the occurrence of coherent states in the place fields.
\begin{definition}(Boltzmann and quenched averages)
Given a function $f(\boldsymbol{s})$, depending on the neuronal configuration $\bb s$, 
the Boltzmann average, namely the average over the distribution \eqref{BGmeasure}, is denoted as $\omega( f(\boldsymbol{s}) )$ and defined as
$$
 \ \omega(f(\boldsymbol{s}))= \frac{\sum_{\{s \}}^{2^N} f(\boldsymbol{s}) e^{-\beta H_N(\boldsymbol{s}|\boldsymbol{\eta})}}{\sum_{\{s \}}^{2^N} e^{-\beta H_N(\boldsymbol{s}|\boldsymbol{\eta})}}.
$$
Further, given a function $g(\bb \eta)$ depending on the realization of the $K$ maps, we introduce the quenched average, namely the average over the realizations of the maps, that is denoted as $\mathbb E_\eta [ g(\bb \eta) ]$ or as $\langle g(\bb \eta) \rangle_{\eta}$ according to the context, and it is defined as
\begin{align}\label{eq:mapsexpectation}
    \mathbb E_\eta [ g(\bb \eta) ] \equiv \avg{g( \bb \eta)}_\eta = \int_{-\pi}^\pi\slr{\prod_{i,\mu=1}^{N,K}  \frac{d\theta^\mu_i}{2\pi}} g(\bb \eta(\bb \theta)).
\end{align}
In the last term we highlighted that maps are determined by the set of angles $\bb \theta$, see \eqref{eq:etas}.
This definition of the quenched average follows from the assumption that the maps are statistically independent and therefore the expectation over the place fields factorizes over the sites $i=1,..,N$ and over the maps $\mu=1,..,K$.
\newline
Finally, we denote with the brackets $\langle \cdot \rangle$ the average over both the Boltzmann-Gibbs distribution and the realization of the maps, that is
$$
\langle \cdot \rangle = \mathbb E_\eta [\omega(\cdot )].
$$
\end{definition}
It is worth deriving some relations that will become useful in the following, in particular, we evaluate the quenched average of a function $g(\bb \eta)$ whose dependence on $\bb \eta$ occurs via the scalar product $\eta^{\mu}_i \cdot a$, where $a$ is a two-dimensional vector with module $|a|$ and direction specified by the versor $\hat a$, that is, $a = |a| \:\hat a$. Then, dropping the scripts $\mu$ and $i$ in $\eta^{\mu}_i$, without loss of generality, we get\footnote{These relations can be easily derived by applying the change of variable $t=\cos \theta$ in the integrals, where $\theta$ is the angle between the two vectors involved in the scalar product, and using $|\eta| = 1$. Also notice that, because of the identity $\frac{1}{\pi}\int_{-1}^1 \frac{dt}{\sqrt{1-t^2}} = 1$, we have $\avg{ \exp(\eta \cdot a) }_\eta \sim 1 + \frac{|a|^2}{4} + \mathcal O(|a|^4)$, for $|a| \to 0$.}
\begin{align}
    &\avg{ g(\eta \cdot a) }_\eta = \int_{-\pi}^\pi \frac{d\theta}{2\pi} \:g(|a| \cos \theta) = \frac{1}{\pi}\int_{-1}^1 \frac{dt}{\sqrt{1-t^2}} \: g(|a| \:t),\label{eq:exp1}\\
    &\avg{ \eta ~ g(\eta \cdot a) }_\eta = \int_{-\pi}^\pi \frac{d\theta}{2\pi} \: \hat a\cos\theta \: g(|a| \cos \theta) = \frac{\hat a}{\pi}\int_{-1}^1 \frac{dt}{\sqrt{1-t^2}} \: t\: g(|a|\: t).\label{eq:exp2}
\end{align}

The relations provided in eq. \eqref{eq:etas}, \eqref{eq:mapsexpectation} and \eqref{eq:exp1}-\eqref{eq:exp2} are valid for the $2$-dimensional unitary circle but can be easily generalized in $d$-dimension. Then, the given map $\eta^\mu$ is a unit vector on the (hyper-)sphere $S^{d-1}$ as $\eta^\mu \in \mathbb R^d, \:\: |\eta^\mu|=1$,
and, in spherical coordinates, $\eta^\mu$ is a function of the angles $\mathbf \Omega = (\theta, \phi, ..)$: $\mathbf \eta^\mu_i = \mathbf \eta^\mu_i(\mathbf \Omega^\mu_i)$. Notice that the dot product of two maps, \emph{i.e.} $\eta^\mu_i \cdot \eta^\mu_j = \cos\gamma$, is still a function of the relative angle $\gamma$ between the two unit vectors, hence our requirement for the kernel function of the model is respected also in $d$-dimension.\\
Then, we introduce the volume form $d\omega_d$ on $S^{d-1}$, which in spherical coordinates can be written as
\begin{align}\label{eq:vform}
    d\omega_d = (\sin\theta)^{d-2} \:d\theta d\omega_{d-1}, \:\:\theta\in[0,\pi].
\end{align}
The expectation over the maps \eqref{eq:mapsexpectation} can therefore be generalized in $d-$dimensions as follows
\begin{align}
    \avg{g( \bb \eta)}_{\eta \in S^{d-1}} = \int \slr{\prod_{i,\mu=1}^{N,K}  \frac{d\omega^\mu_i}{|S^{d-1}|}} g(\bb \eta(\bb \omega)),
\end{align}
where the normalization factor $|S^{d-1}|$ is the volume of the sphere, which is computed by integrating \eqref{eq:vform}, and reads
\begin{align*}
    |S^{d-1}| =\int d\omega_d = \frac{2\pi^\frac{d}{2}}{\Gamma\lr{\frac{d}{2}}},
\end{align*}
where $\Gamma$ is the gamma function. Further, the relations \eqref{eq:exp1}-\eqref{eq:exp2} can be generalized as
\begin{align}
    &\avg{ g(\eta \cdot a) }_\eta = \frac{1}{|S^{d-1}|}\int d\omega_d \:g(\eta \cdot a) = \Omega_d \int_{-1}^1 dt\:\lr{1-t^2}^\frac{d-3}{2} \: g(|a| \:t),\label{eq:exp1_d}\\
    &\avg{ \eta ~ g(\eta \cdot a) }_\eta = \hat a\:\Omega_d \int_{-1}^1 dt\: t\lr{1-t^2}^\frac{d-3}{2} g(|a|\: t),\label{eq:exp2_d}
\end{align}
where we performed the change of variables $t=\cos\theta$ (with $\theta$ being the angle between $a$ and $\eta$), after which and the factor $\Omega_d$ emerges as:
\begin{align}
    \Omega_d = \frac{|S^{d-2}|}{|S^{d-1}|} = \frac{\Gamma \lr{\frac{d}{2}}}{\sqrt \pi \:\Gamma \lr{\frac{d-1}{2}}}.
\end{align}
Notice that for $d=2$ one has $\Omega_2 = \frac{1}{\pi}$, restoring the relations \eqref{eq:exp1}-\eqref{eq:exp2}.
Notice also that the series expansion of \eqref{eq:exp1_d} reads $\avg{ \exp(\eta \cdot a) }_{\eta \in S^{d-1}} \sim 1 + \frac{|a|^2}{2d} + \mathcal O(|a|^4)$.

\begin{definition}\label{selfaverage}(Replica symmetry)
We assume that, in the thermodynamic limit $N \to \infty$, all the order parameters self-average around their mean values, denoted by a bar, that is
\begin{eqnarray}
\lim_{N \to \infty} \mathbf P(x_\mu) &=& \delta\left(x_\mu  -\overline{x}_{\mu} \right), \:\:\forall \mu \in (1,...,K),\\
\lim_{N \to \infty}\mathbf P(m) &=& \delta\left(m -\overline{m} \right),\\
\lim_{N \to \infty} \mathbf  P(q_{12}) &=& \delta\left(q_{12} - \overline q_{2} \right),\\
\lim_{N \to \infty} \mathbf  P(q_{11}) &=& \delta\left(q_{11} - \overline q_{1} \right).
\end{eqnarray}
This assumption is the so-called {\em replica symmetric approximation} in spin-glass jargon, see e.g., \cite{Barra-JSP2008,GuerraNN} or {\em concentration of measure} in probabilistic vocabulary \cite{talagrand2003spin}.
\end{definition}


\begin{definition}(Storage capacity)
The storage capacity of the network is denoted as $\alpha$ and defined as
\begin{align}
    \alpha = \frac{K}{N}
\end{align}
The regime where the number of stored charts scales sub-linearly with the network size, i.e. where $\alpha=0$, is referred to as low-storage, while the regime where the number of stored charts scales linearly with the network size, i.e. where $\alpha>0$, is referred to as high-storage.
\end{definition}
The storage capacity $\alpha$, along with the noise level $\beta$ and the inhibition parameter $\lambda$, constitute the control parameters of the system under study and, by tuning their values, the behavior of the (mean values of the) order parameters $\overline{x}, \overline{m}, \overline{q}$ changes accordingly. 
\newline
To quantify their evolution in the $(\alpha,\beta,\lambda)$ space, we further introduce the main quantity for our investigation, that is
\begin{definition}(Free energy)
The intensive free-energy of the Battaglia-Treves model equipped with McCulloch $\&$ Pitts neurons is defined as
\begin{equation}\label{eq:Adefinition}
    A_{N,K}(\beta,\lambda)= \frac{1}{N}\mathbb{E}_{\eta}\ln  Z_{N}(\beta, \lambda, \bb \eta)=\frac{1}{N}\mathbb{E}_{\eta}\ln \sum_{\{s\}}^{2^N}\exp\left(-\beta H_N(\boldsymbol{s}|\boldsymbol{\eta})\right)  \end{equation}
and, in the thermodynamic limit, we write
\begin{equation}\label{eq:Adefinition}
    A(\alpha,\beta,\lambda)=\lim_{N\to \infty} A_{N,K}(\beta,\lambda).
\end{equation}
\end{definition}
The explicit knowledge of the free energy in terms of the control parameters $\alpha, \beta, \lambda$ and of the (expectation of the) order parameters $\overline{x}, \overline{m}, \overline{q}_1, \overline{q}_2$ is the main focus of the present investigation. In fact, once its explicit expression is obtained, we can extremize the free energy w.r.t. the order parameters to obtain a set of self-consistent equations for their evolution in the space of the control parameters: the study of their solutions allows us to paint the phase diagram of the model and thus to know {\em a priori} in which regions in the $(\alpha, \beta, \lambda)$ space, the charts can be successfully retrieved by the place cells as we now discuss.

\subsection{Phase diagram of the Battaglia-Treves model with McCulloch $\&$ Pitts neurons}\label{Sec:2.3} 

We recall that the quenched average of the free energy for the standard Battaglia-Treves model has been computed in the thermodynamic limit by using the replica trick at the replica symmetric level of approximation, see e.g.  \cite{battaglia1998attractor,MonassonPlaceCellsPRL}, and a purpose of the present paper is to recover the same expression of the free energy for the current version of the model (where neurons are binary) by adopting the (mathematically more-controllable) Guerra interpolation: in Sec.~\ref{LowStorage} we address the simpler low-storage regime, by adapting to the case the Hamilton-Jacobi approach \cite{HJ-Barra2010Aldo,HJ-Barra2013Gino}  (see also the works by Mourrat and Chen for a sharper mathematical control \cite{HJ-Chen22,HJ-Mourrat21,HJ-Mourrat22,HJ-Mourrat23book}) and in Sec.~ \ref{Guerra} we address the more challenging high-storage regime, by adapting to the case the standard interpolation technique based on stochastic stability \cite{glassy,GuerraNN,guerra_broken}; in App.~\ref{Replicas} we perform the evaluation of the quenched free energy also via the replica trick, finding overall agreement among the results obtained with the various approaches. 
\newline
Whatever the route, the main result  can be summarized by the next
\begin{theorem}
In the thermodynamic limit, the replica-symmetric quenched free-energy of the Battaglia-Treves model, equipped with McCulloch-Pitts neurons as defined in eq. \eqref{def:cost_function}, for the $S^{d-1}$ embedding space, can be expressed in terms of the (mean values of the) order parameters $\overline{m}, \overline{x}, \overline{q}_1, \overline{q}_2$ and of the control parameters $\alpha, \beta, \lambda$, as follows
\begin{eqnarray}\label{eq:ARS}
    A^{RS}(\alpha,\beta,\lambda) &=& - \frac{\beta}{2} (1-\lambda) \overline m^2 - \frac{\beta}{2} \overline x^2 - \frac{\alpha \beta}{2} \frac{\overline q_1 - \frac{\beta}{d}(\overline q_1-\overline q_2)^2}{\lr{1-\frac{\beta}{d} (\overline q_1-\overline q_2)}^2} - \frac{\alpha d}{2}\ln\lr{1-\frac{\beta}{d}(\overline q_1-\overline q_2)} +   \\  \nonumber
   &+&\frac{\alpha\beta}{2}\frac{\overline q_2}{1-\frac{\beta}{d}(\overline q_1-\overline q_2)}   + \mathbb E_\eta \int Dz \ln\lr{ 1 + \exp\lr{\beta(1-\lambda) \overline m +\beta \overline x \cdot \eta + \beta \frac{ \frac{\alpha}{2} + \sqrt{\frac{\alpha \overline q_2}{d}} \:z}{1-\frac{\beta}{d}(\overline q_1-\overline q_2)} }},
\end{eqnarray}
where the superscript $RS$ reminds us that $A^{RS}(\alpha,\beta,\lambda)$ is the replica symmetric approximation of the true free energy\footnote{The point here is that these type of neural networks are spin glasses in a statistical mechanical jargon, hence Parisi's replica symmetry breaking is expected to occur in the low noise and high storage limit \cite{Amit}. Yet, in this first investigation we confine ourselves in providing an exhaustive picture of the RS scenario, that is usually the first to be formalized.}.
The mean values $\overline m$, $\overline q_1$, $\overline q_2$, $\overline x$ appearing in the above expression have to extremize the free-energy, which implies that their values are obtained as solutions of the constraint  $\nabla_{\overline{x},\overline{m},\overline q_1,\overline q_2} A^{RS}(\alpha,\beta,\lambda)=0$, that yields the following self-consistency equations 
\begin{align} 
    \label{eq:tre}
    &\langle x \rangle = \overline x = \int d\mu(z) \: \avg{\eta \:\sigma\lr{ \beta h (z)}}_\eta,\\
    \label{eq:uno}
    &\langle q_{11}\rangle = \overline q_1 = \overline m = \int d\mu(z) \: \avg{\sigma\lr{ \beta h(z) }}_\eta,\\
    \label{eq:due}
    &\langle q_{12}\rangle = \overline q_2 = \int d\mu(z) \: \avg{\sigma^2\lr{ \beta h (z)}}_\eta,
\end{align}
where $\sigma(t) = \frac{1}{1+e^{-t}}$ is the sigmoid function, $d\mu(z)$ is the Gaussian measure for $z\sim \mathcal N(0,1)$, and $h(z)$ represents the internal field acting on neurons and reads as
\begin{align}
    h(z) = (1-\lambda) \overline m + \overline x \cdot \eta + \frac{ \frac{\alpha}{2} + \sqrt{\frac{\alpha \overline q_2}{d}} \:z}{1-\frac{\beta}{d}(\overline q_1-\overline q_2)} .
\end{align}
\end{theorem}
A proof of this theorem, based on Guerra's interpolation, can be found in Sec.~\ref{Guerra}, further, in App.~\ref{Replicas}, we report an independent derivation based on the usage of the replica trick. Finally, we note that  the low storage solution provided in Sec.~\ref{LowStorage} can be obtained simply by setting $\alpha=0$ in the above expression \eqref{eq:ARS}.
\mycomment{
\begin{figure}[tb]    
    \centering
    \includegraphics[width=9.5cm]{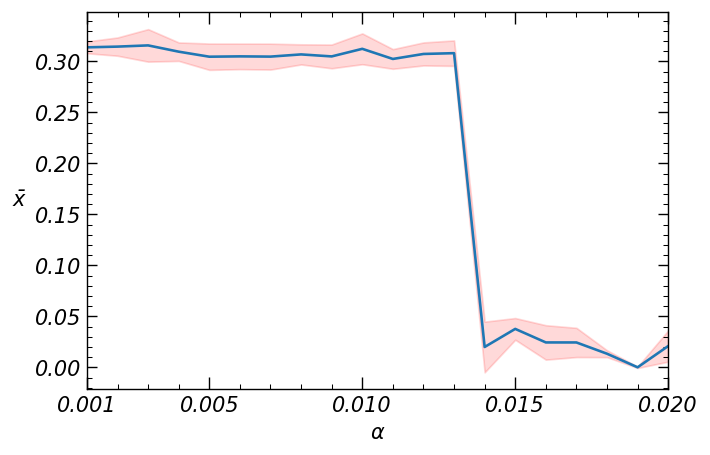}
    \caption{WORK IN PROGRESS\\
    The numerical solution of the self-consistency equation at $\beta = 100$ for the population vector $\overline x$ at various values of the load $\alpha$ is shown. The theoretical maximum of $\overline x = \frac{1}{\pi}$ is reached for lower values of $\alpha$ as expected, indicating the formation of a coherent state with a bump of size $\delta\sim \pi$ in the retrieved place field; the system undergoes a first order phase transition for $\alpha_c \sim 0.013 \div 0.014$ at $\beta = 100$. The pink bands show the uncertainty of the numerical solution.}
\end{figure}
}

\begin{corollary}
In order to numerically solve the self-consistent equations \eqref{eq:tre}-\eqref{eq:due}, it is convenient to introduce the quantity $C=\frac{\beta}{d}(\overline q_1 - \overline q_2)$, make the change of variables $(\overline x, \overline q_1, \overline q_2) \to (\overline x, C, \overline q_2)$ and write them as
\begin{align}
    \label{SelfHS3} 
    &\overline x = \Omega_d \int Dz \int_{-1}^1 dt\: t   \lr{1-t^2}^\frac{d-3}{2}\: \sigma\lr{ \beta h(z,t)},\\
    &\overline q_2 = \Omega_d\int Dz \int_{-1}^1dt\:\lr{1-t^2}^\frac{d-3}{2}\: \sigma^2\lr{ \beta h(z,t)}\\
    \label{SelfHS4}
    &C \equiv \frac{\beta}{d}(\overline q_1-\overline q_2) =  \Omega_d \frac{1-C}{\sqrt{\alpha \overline q_2 d}}\int Dz \:z\: \int_{-1}^1dt\:\lr{1-t^2}^\frac{d-3}{2}\: \sigma\lr{ \beta h(z,t)},
\end{align}
where we posed 
\begin{align}
    &h(z,t) \equiv (1-\lambda) \lr{\frac{d}{\beta} C + \overline q_2} + t \overline x  + \frac{ \frac{\alpha}{2} + \sqrt{\frac{\alpha \overline q_2}{d}} \:z}{1-C}.\label{eq:hfield}
   \end{align}
\end{corollary}
\begin{proof}
The self-consistent equations \eqref{eq:tre}-\eqref{eq:due} can be made explicit simply by applying \eqref{eq:exp1}-\eqref{eq:exp2} to evaluate the expectation over the map realization. 
\end{proof}

The self-consistent equations \eqref{SelfHS3} - \eqref{SelfHS4}, with the internal field defined according to \eqref{eq:hfield}, are solved numerically to draw a phase diagram for the model as reported in Fig.~\ref{fig:phasediagram}. In particular, we find that the model exhibits three different phases. In the \emph{paramagnetic} phase (PM) the model does not retrieve any map, hence $\overline x_\mu = 0$, for any $\mu=1,..,K$, and the replicas are totally uncorrelated: 
in this regime noise prevails over signals and network's dynamics is ergodic. However, lowering the noise, the ergodic region breaks and the network may enter a spin glass region (if the amount of stored maps is too large) or a coherent state, namely a retrieval region where charts are spontaneously reinstated by the network.
\newline 
In order to check these analytical findings (in particular, to inspect the goodness of the RS approximation and possible finite-size effects for this model) we run extensive Monte Carlo (MC) simulations obtaining a very good agreement between theory and numerical simulations, as shown in Fig.~\ref{fig:MC}. 
\newline
Under the RS approximation, the maximum capacity of the model lies  in the $\beta \to \infty$ limit and its value can be estimated numerically by studying that limit of the self-consistency equations, which leads to the following
\begin{figure}[tb] 
    \centering
    \includegraphics[width=8.25cm]{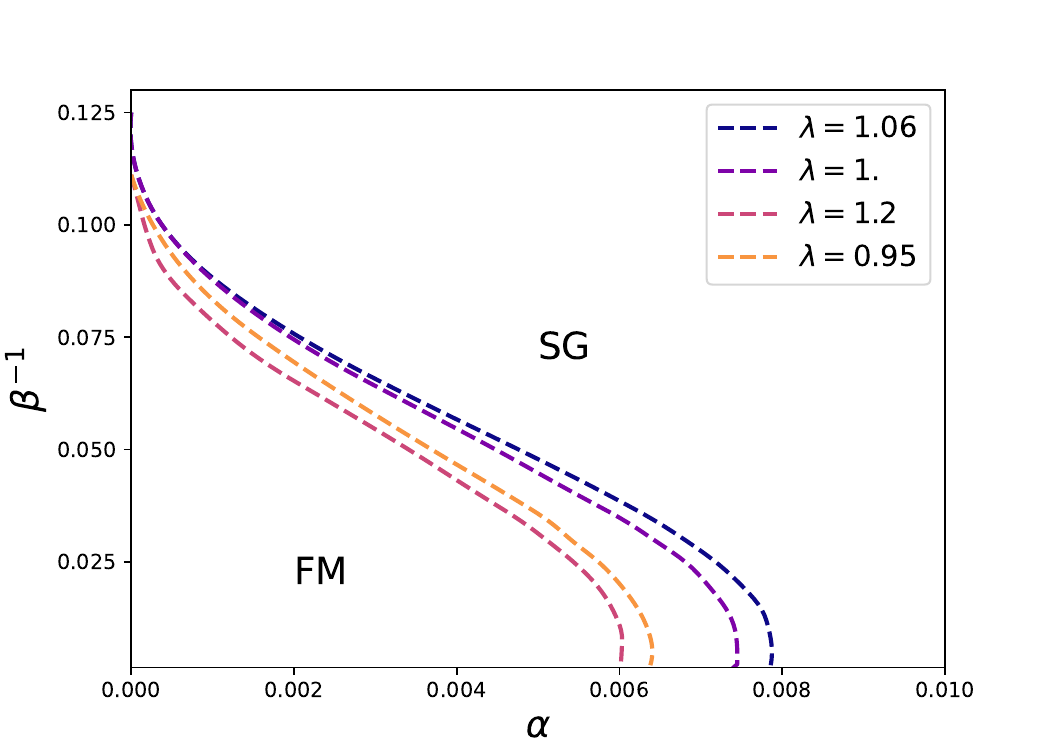}
    \includegraphics[width=8.25cm]{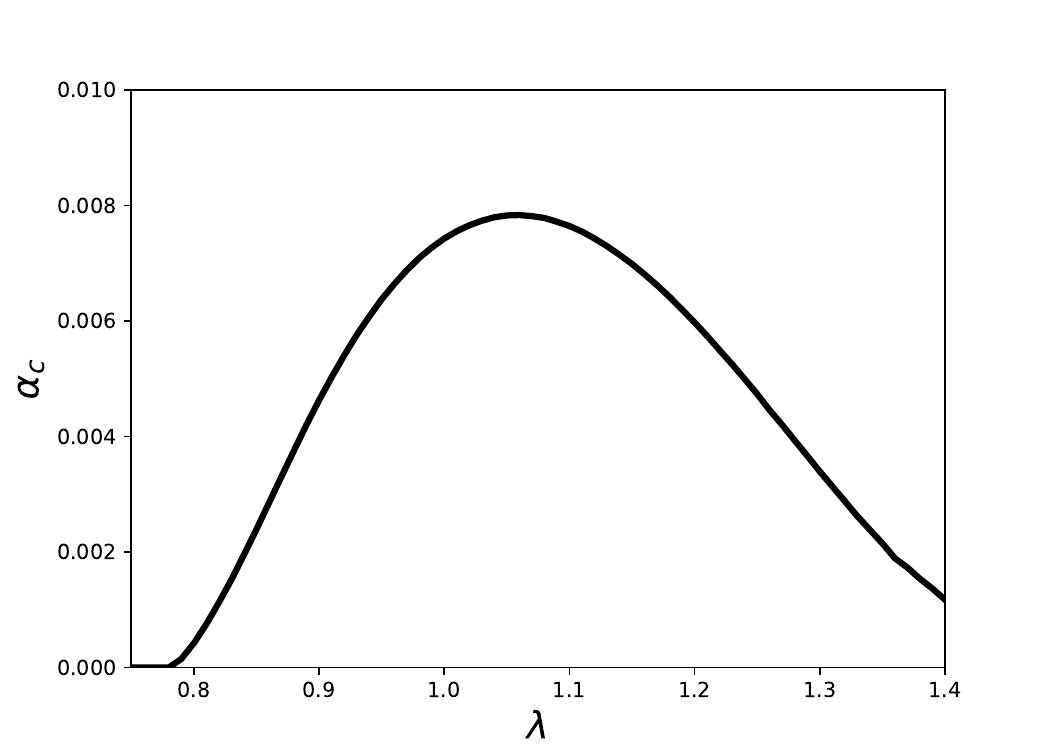}
    \caption{\emph{Left}: Lower region of the phase diagram of the model in $d=2$ dimensions in the load $\alpha$ and noise $\beta$ plane, at various inhibition strength $\lambda$, as shown in the legend. We recognize the presence of two phases, the Spin Glass phase (SG), where $\overline x=0$, $(\overline q_1)^2\ll \overline q_2\leq \overline q_1$, and the Ferromagnetic phase (FM), with $\overline x>0$, $\overline q_2= \overline q_1$ (i.e. the region where coherent states collectively appear). Note that the paramagnetic phase, which would be in the upper region of the diagram,  is not shown for the sake of clearness, to emphasize the boundary between the SG and FM phases, which differs from what is seen in  neural networks with linear threshold units \cite{Ale3,Ale4}).
    The {\em spin glass phase} (SG) is characterized by the absence of any coherent state in any map, hence one still has $\overline x = 0$, but one observe the onset of correlation among replicas, which, for a given activity level $\overline q_1=\overline m$, is characterized by a value of $\overline q_2$ (that lies between $\overline m^2$ and $\overline m$), which is significantly higher than $\overline m^2$ (namely $\overline m^2 \ll \overline q_2 \leq \overline m$), with the particular case $\overline q_2 = \overline m$ depicting a frozen configuration of the network (that however is not a coherent state in any map $\eta_\mu$).\\
    In the {\em ferromagnetic phase} (FM) the neural network is able to retrieve a coherent state in one of its maps (which one depends on the initial conditions of the dynamics), say $\eta^1$, hence, for a given activity level $\overline q_1=\overline m$, one measures $\overline x_1 > 0$. In this regime the configuration of the network is obviously frozen, therefore the overlap $\overline q_2$ takes its highest value: $\overline q_2 \sim \overline m$, with $\overline m - \overline q_2 = \frac{Cd}{\beta}\to0$ in the $\beta\to\infty$ limit.  The phase transition is discontinuous (first order).
    \emph{Right}: The critical load of the model as a function of the global inhibition strength $\lambda$; the maximum value of $\alpha_c \sim 0.0078$ is found at $\lambda_{\alpha_c} = 1.06$, {\em close to criticality}, namely $\lambda_{\alpha_c} \sim 1$. 
    \label{fig:phasediagram}}
\end{figure}
\begin{figure}[t]
\centering
    \begin{minipage}[t]{.32\textwidth}
        \centering
        \includegraphics[width=\textwidth]{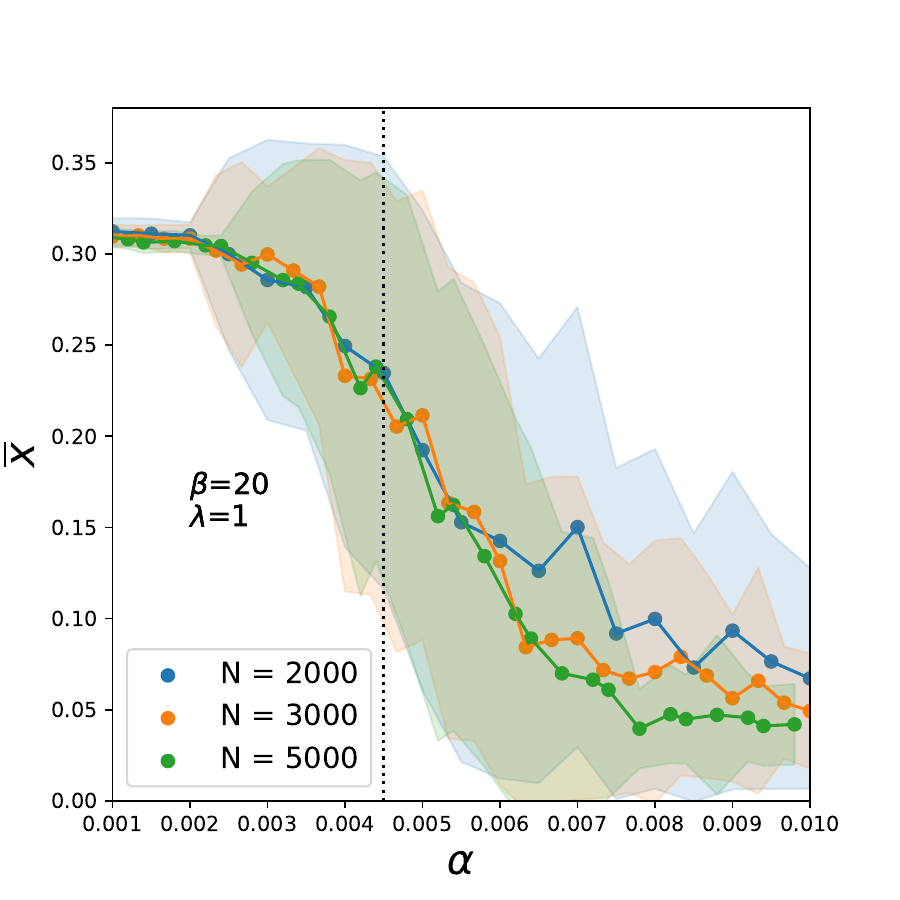}
        \subcaption{}
        \label{fig:MC_a} 
    \end{minipage}
    \begin{minipage}[t]{.32\textwidth}
        \centering
        \includegraphics[width=\textwidth]{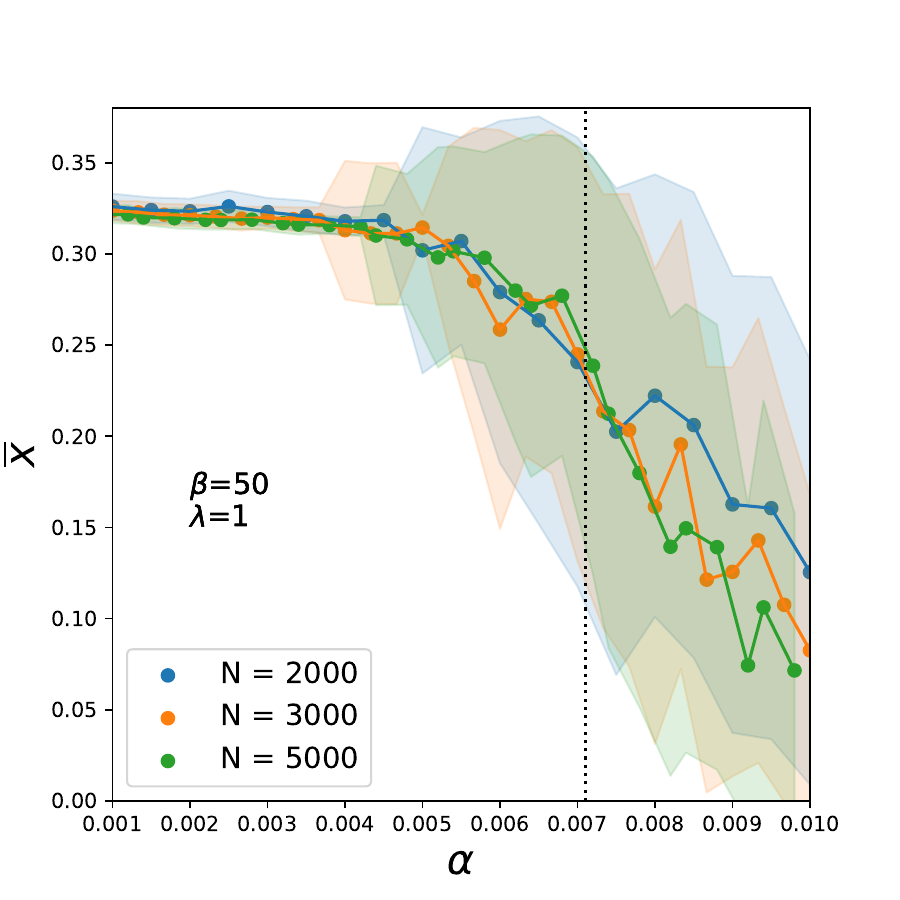}
        \subcaption{}
        \label{fig:MC_b}
    \end{minipage} 
    \begin{minipage}[t]{.32\textwidth}
        \centering
        \includegraphics[width=\textwidth]{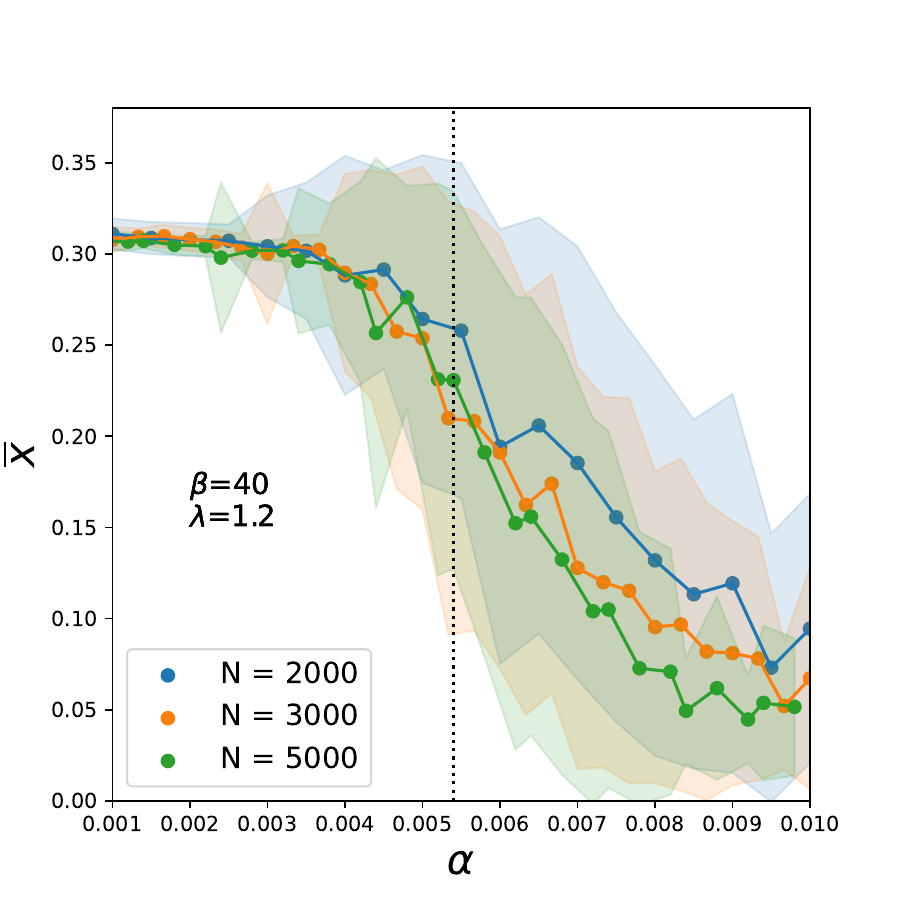}
        \subcaption{}
        \label{fig:MC_c}
    \end{minipage}
    \caption{Finite Size Scaling simulations of neural dynamics (Monte Carlo runs), for three different choices of the control parameters: $\beta = 20$ and $\lambda=1$ \emph{(a)}, $\beta = 50$ and $\lambda=1$ \emph{(b)}, $\beta = 40$ and $\lambda=1.2$ \emph{(c)}. Each scenario has been investigated by enlarging the network size from $N=2000$ to $N=5000$: the error bands (indicated by the shadowed regions) set to $\pm 1$ standard deviation around the experimental points. Each MC run is the average over $50$ samples. The dotted line indicates the theoretical transition line, as shown in the phase diagram provided Fig. \ref{fig:phasediagram} and they sit sharply at the inflection points of $\overline{x}(\alpha)$.}
    \label{fig:MC}
\end{figure}
\begin{corollary}
In the noiseless limit $\beta\to\infty$ limit, the self-consistency equations \eqref{SelfHS3}-\eqref{SelfHS4} in the unitary circle (that is, for a one-dimensional manifold embedded in $d=2$ dimensions), read as
\begin{align}
    &\overline x = \frac{1}{2\pi} \int_0^\pi d\theta\:\cos\theta\:\erf\lr{\frac{g(\theta)}{\sqrt 2}},\label{eq:selfzerot1}\\
    &\overline q_2 = \frac{1}{2} + \frac{1}{2\pi} \int_0^\pi d\theta\:\erf\lr{\frac{g(\theta)}{\sqrt 2}},\label{eq:selfzerot2}\\
    &C=\frac{1-C}{\sqrt{2\pi^3 \alpha \overline q_2 d}}\int_0^\pi d\theta\:\exp\lr{-\frac{g(\theta)^2}{2}},\label{eq:selfzerot3}
\end{align}
where we posed 
\begin{align}
    &g(\theta)=\sqrt{\frac{d}{\alpha \overline q_2}}\slr{\frac{\alpha}{2} + (1-C)\lr{(1-\lambda) \overline q_2 + \overline x \cos\theta}}.\label{eq:gt}
\end{align}
\end{corollary}
\begin{proof}
These equations in the noiseless limit $\beta \to \infty$ are obtained by considering that, in this limit, the sigmoid function reduces to the Heaviside function $\Theta$, \emph{i.e.} $\sigma(\beta h) \overset{\beta \to \infty} \longrightarrow\Theta(h)$. The Gaussian integral appearing in the self-consistency equations \eqref{SelfHS3}-\eqref{SelfHS4} can then be restricted to the domain where the internal field $h(z)$ is positive, that is the interval $h(z)\in [-g(t),\infty)$, where $g(t)$ is defined in eq. \eqref{eq:gt}: after a trivial rescaling of the integrals, we obtain the expressions \eqref{eq:selfzerot1}-\eqref{eq:selfzerot3} in terms of the error function $\erf(x)=\frac{2}{\sqrt \pi} \int_0^x dz\: e^{-z^2}$.\\
Note that, as standard \cite{Amit}, we used $\frac{Cd}{\beta} + \overline q_2 \overset{\beta\to\infty}{\longrightarrow} \overline q_2$. 
\end{proof}
These equations can be solved numerically and, in particular, we find that there exists a critical value for $\alpha$, referred to as $\alpha_c$, such that above that threshold no positive solution for $\overline x$ exists. The numerical estimate of this critical storage $\alpha_c(\lambda)$ as a function of $\lambda$ is reported in Fig. \ref{fig:phasediagram} (right panel); the maximal critical load is found to be $\alpha_c \sim 0.0078$ at $\lambda = 1.06$, suggesting that the network  works at its best  for mild values of global inhibition. 
Indeed, at this value of $\lambda$, the model benefits of a weak asymmetry in the excitation-inhibition trade-off in favor of the inhibitory contribution (while at $\lambda=1$ the excitation contribution in the Hamiltonian exactly matches the inhibitory contribution). 

\begin{figure}[!h]
    \centering
    \includegraphics[width=12cm]{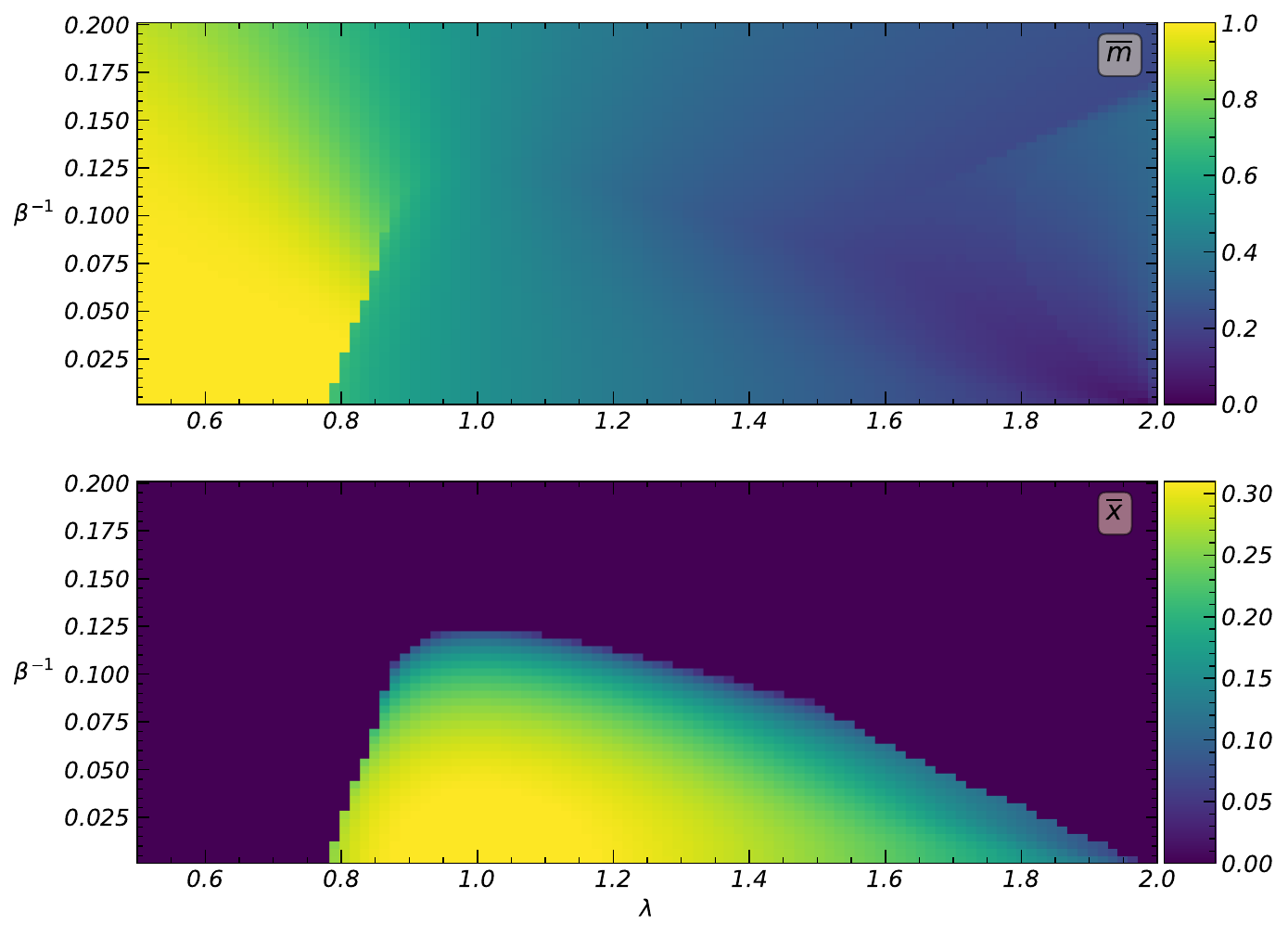}
    \caption{
    Phase diagram of the model in $d=2$ dimensions in terms of the inhibition strength $\lambda$ and noise $\beta$, in the low-storage regime $\alpha = 0$. We shown the overall magnetization of the network $\overline m$ (\emph{upper panel}) and the module of the neural activity $\overline x$ (\emph{lower panel}), as functions of the control parameters $\beta, \lambda$.}
    \label{fig:PD1}  
\end{figure}

\bigskip

In the low-storage regime (i.e., $\alpha = 0$) the phase diagram only depends on two control parameters: the neural noise $\beta$ and the inhibition strength $\lambda$. 
In order to work out a phase diagram for the network in this regime it is enough to simplify the self-consistency equations \eqref{eq:tre}-\eqref{eq:uno} by setting $\alpha=0$ therein, thus obtaining
\begin{align}\label{SC-Low2}
    &\overline x = \frac{1}{\pi}\int_{-1}^1 dt \:\frac{t}{\sqrt{1-t^2}} \: \sigma\lr{ \beta (1-\lambda) \overline m + \beta  t \overline x},\\ \label{SC-Low1}
    &\overline m = \frac{1}{\pi}\int_{-1}^1 \frac{dt}{\sqrt{1-t^2}} \: \sigma\lr{ \beta (1-\lambda) \overline m + \beta  t \overline x}
\end{align}
These equations are solved numerically to trace the evolution of $\overline m$ and $\overline x$ in the $\beta, \lambda$ plane and results are presented in Fig.~\ref{fig:PD1}: as expected, and as anticipated in the previous section, the limits $\lambda \ll 1$ and $\lambda \gg 1$ both lead to ``paramagnetic'' phases where $\overline x=0$ and the computational capabilities of the network are lost. In constrast, in the region close to $\lambda=1$ a retrieval phase with $\overline x>0$ and $\overline m\sim 0.5$ naturally appears.
 

\mycomment{
\begin{align}\label{SelfHS1}
    &\overline m = \frac{1}{\pi}\int Dz \int_{-1}^1\frac{dt}{\sqrt{1-t^2}} \: \sigma\lr{ \beta \lr{(1-\lambda) \overline m + t \overline x  + \frac{ \frac{\alpha}{2} + \sqrt{\frac{\alpha q_2}{d}} \:z}{1-C} }},\\ \label{SelfHS2}
    &q_2 = \frac{1}{\pi}\int Dz \int_{-1}^1\frac{dt}{\sqrt{1-t^2}} \: \sigma^2\lr{ \beta \lr{(1-\lambda) \overline m + t \overline x  + \frac{ \frac{\alpha}{2} + \sqrt{\frac{\alpha q_2}{d}} \:z}{1-C} }},\\ \label{SelfHS3}
    &C = \frac{1-C}{\sqrt{\alpha q_2 d}}\frac{1}{\pi}\int Dz \:z\: \int_{-1}^1\frac{dt}{\sqrt{1-t^2}} \: \sigma\lr{ \beta \lr{(1-\lambda) \overline m + t \overline x  + \frac{ \frac{\alpha}{2} + \sqrt{\frac{\alpha q_2}{d}} \:z}{1-C} }},\\ \label{SelfHS4}
    &\overline x = \frac{1}{\pi}\int Dz \int_{-1}^1 dt \:\frac{t}{\sqrt{1-t^2}} \: \sigma\lr{ \beta \lr{(1-\lambda) \overline m + t \overline x  + \frac{ \frac{\alpha}{2} + \sqrt{\frac{\alpha q_2}{d}} \:z}{1-C} }}.
\end{align}
}

\subsection{The emergence of map-selective cells}\label{Citua}
We note that an equivalent mathematical model can be given an alternative implementation in terms of a network of  highly selective  hidden variables (firing if and only if the rat transits within a given place field). These cells are not the standard neurons $\boldsymbol{s}$ we introduced, as those collectively account for (several) maps and none of them is specifically firing for a given place field. Indeed the same formal model admits a dual representation of its Hamiltonian where these highly-selective neurons naturally appear as chart cells\footnote{One may question that the hidden neurons are no longer Boolean variables,  but Gaussian ones. While this is certainly true, models with hidden layers equipped with binary neurons have been studied too and qualitatively the {\em duality of representation} picture remains robust \cite{DaniPRE2018}}:
\newline
in this respect, it is enough to point out that the  partition function of the model can be written as
\begin{eqnarray}
    Z_N(\beta,\lambda,\boldsymbol{\eta}) &=& \sum_s \exp\left(\frac{\beta}{2N} \sum_{i,j}^{N,N}\sum_{\mu=1}^K \eta_i^{\mu}\eta_j^{\mu} s_i s_j -\frac{\beta(\lambda-1)}{2N} \sum_{i,j}^{N,N}s_i s_j \right)  \\ \label{duality}
    &=& \sum_s  \int d\mu(z_{\mu})\exp\left( \frac{\beta}{\sqrt{N}}\sum_i^N \sum_{\mu}^K \eta_i^{\mu}s_i z_{\mu}-\frac{\beta(\lambda-1)}{2N} \sum_{i,j}^{N,N}s_i s_j \right),
\end{eqnarray}
where, in the last line, we introduced $K$ hidden variables that could play the formal role of chart cells $z_{\mu},\ \mu \in (1,...,K)$, one per chart, such that the exponent in eq.~\eqref{duality} can be thought of as an effective Hamiltonian as stated by the next
\begin{proposition}
Once introduced $K$ real-valued  $z_{\mu}, \ \mu \in (1,...,K)$, hidden neurons equipped with a standard Gaussian prior (i.e. the measure $d\mu(z_{\mu})$), the Battaglia-Treves model admits a dual representation in terms of a bipartite neural network whose cost function reads as    
\begin{equation}\label{placecell}
H_{N,K}(s,z|\boldsymbol{\eta})= -\frac{1}{\sqrt{N}}\sum_i^N \sum_{\mu}^K \eta_i^{\mu}s_i z_{\mu}+\frac{\lambda-1}{2N} \sum_{i,j}^{N,N}s_i s_j= -\sqrt{N}\sum_{\mu}^K x_{\mu}z_{\mu} + (\lambda-1) N m^2.
\end{equation} 
\begin{figure}[!h]
    \centering
    \includegraphics[width=8cm]{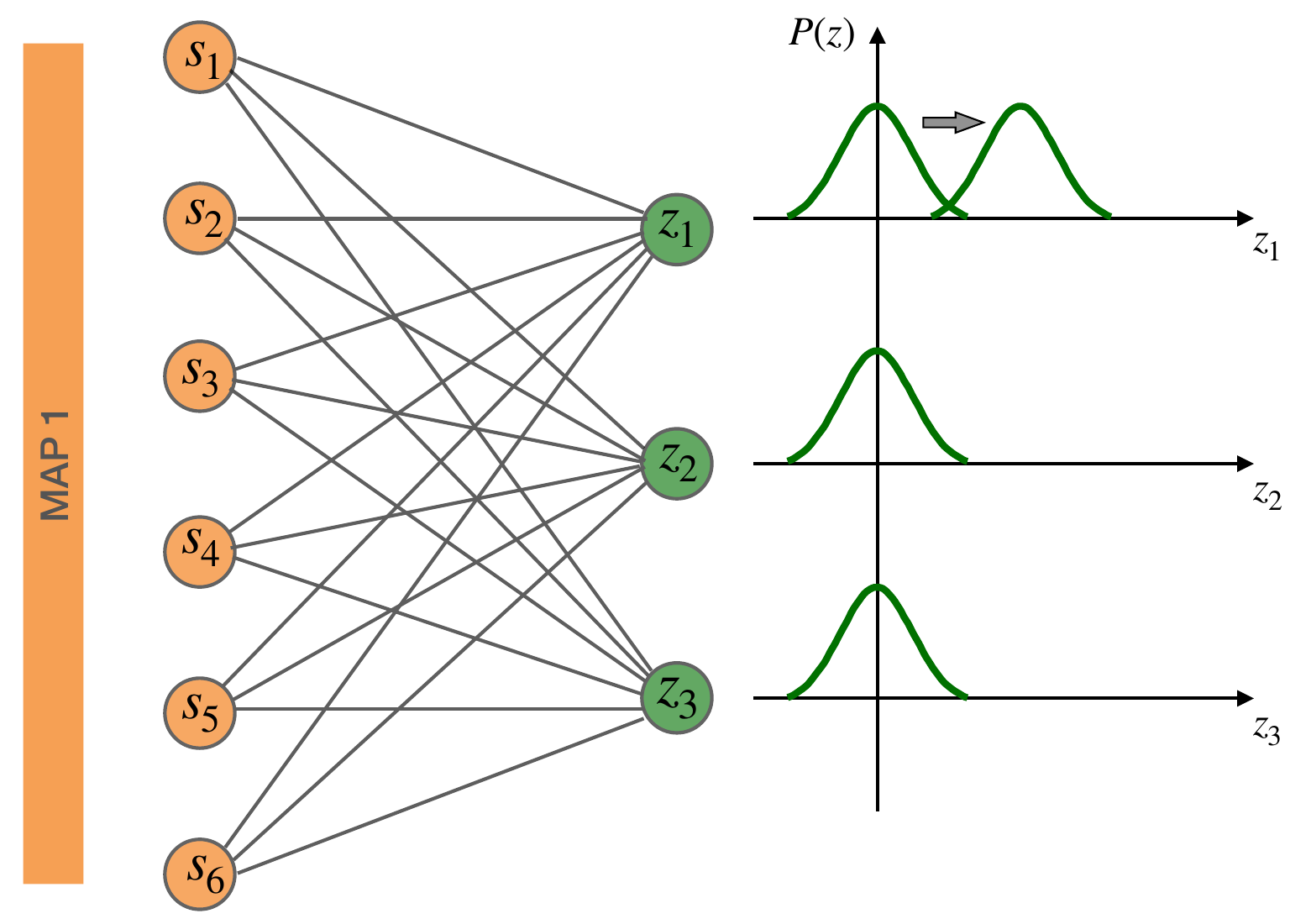}
    \caption{
Dual representation of the Battaglia-Treves model in terms of a bipartite network equipped with grandmother place cells as hidden neurons: these are one per stored chart, such that the retrieval of the -say- first map $\xi^{1}$ by the $\boldsymbol{s}$ triggers the corresponding $z_1$ cell to fire too.}
    \label{fig:double}  
\end{figure}
Note that  these chart-selective neurons  $\boldsymbol{z}$ indicate individual charts: their mean values   are zero unless the conjugated order parameters $x_{\mu} \neq  0$ for some $z_\mu$: in this case, as the overall field experienced by the generic $z_{\mu}$ is proportional $\sqrt{N} x_{\mu}$,   $x_{\mu}>0$ this triggers a response in $z_{\mu}$, as shown in Figure~\ref{fig:double}.
\end{proposition}
\begin{remark}
The stochastic dynamics of such a model can be easily simulated by computing the distributions $P(\mathbf z| \mathbf s)$ and $P(\mathbf s | \mathbf z)$, which can be explicitly evaluated  given the factorized structure of the cost function in the chart cells and read as
\begin{align}
\label{eq:rbm_1}
    &P(\mathbf z| \mathbf s) = \prod_\mu^K \mathcal N(M_\mu + \sigma_\mu^2 \beta \sqrt N x_\mu(\mathbf s), \mathbb I_d \:\sigma_\mu^2),\\
    \label{eq:rbm_2}
    &P(\mathbf s = 1 | \mathbf z) = \prod_i^N \sigma \lr{\frac{\beta}{\sqrt N} \sum_\mu \eta^\mu_i \cdot z_\mu + \beta h_i},
\end{align}
where $M_{\mu}$ accounts for an eventual external stimulus pointing toward a specific chart.
\end{remark}


\section{Guerra's techniques for place cells}\label{sec:trois}
In this section we develop the underlying mathematical methodology and we provide two statistical mechanical techniques, both based on Guerra's interpolation \cite{guerra_broken}, to solve for the free energy of the model under study: the former, meant to address the low-load regime, is the Hamilton-Jacobi approach and it is based on the so-called {\em mechanical analogy} \cite{HJ-Barra2010Aldo,HJ-Barra2013Gino}. This analogy lies in the observation that the free energy in Statistical Mechanics plays as an action in Analytical Mechanics, namely it obeys a Hamilton-Jacobi PDE in the space of the couplings: as a result, the explicit expression for the free energy can be obtained as a solution  of the Hamilton-Jacobi equation, that is by leveraging tools typical of analytical mechanics rather than statistical mechanics. The latter, meant to cover the high-load regime, is the standard one-body interpolation based on stochastic stability  \cite{GuerraNN,Fachechi1}: here we interpolate between the free energy of the original model and the free energy of a suitably-designed one-body  model whose solution is straightforward as one-body models enjoy trivial factorized probability distributions. To connect these two extrema of the interpolation we use the fundamental theorem of calculus: we calculate the derivative of the interpolating free energy and we integrate it over the interval of the interpolating parameter such that the bounds correspond to the original and one-body free energies and, crucially, the assumption of replica symmetry makes the integral analytical.
\newline
These techniques are exploited in Secs. \ref{LowStorage} and \ref{Guerra}, respectively.

\subsection{Phase diagram in the map low-storage regime}\label{LowStorage}
We start our analytical investigation by addressing the low-storage regime -where no spin-glass know-how is required- and we exploit the {\em mechanical analogy}. We first introduce three interpolating parameters: a variable $t \in \mathbb{R}^+$ to mimic {\em time} and three {\em spatial coordinates} $(y,\vec{z}) \in \mathbb{R} \times \mathbb{R}^2$ that we use to define an interpolating free energy $A(t,y,\vec{z})$\footnote{Note that $\vec{z}=(z_1, z_2)$ is a bi-dimensional vector in the present two-dimensional embedding of a one-dimensional manifold but, should we work in $K$ dimensions, $z$ would be a $K$ dimensional vector.}. Then, once checked that $A(t=\beta,y=0,\vec{z}=0)$ coincides with $A(\beta,\lambda)$, we show that $A(t,y,\vec{z})$ plays as an action in this spacetime, namely that it obeys a Hamilton-Jacobi PDE. Finally, we find the PDE solution by integrating the Lagrangian coupled to this action over time which, in turn, provides an explicit expression also for the original free energy of the Battaglia-Treves model confined to the low-storage.

As the low-storage regime is defined by $\alpha=0$, we consider the case where the model only stores and retrieves one map, say $\eta^1$ with no loss of generality \cite{Coolen,Longo}: in this setting, the Hamiltonian reduces to that provided in the next
\begin{definition}(One-map Cost function)
The Battaglia-Treves cost function \eqref{eq:H_Hop} for map's recognition, equipped with  $N$ McCulloch-Pitts neurons $\boldsymbol s = (s_1,...,s_N) \in \{0,1\}^N$ and a free parameter $\lambda \in \mathcal{R}^+$ to tune the global inhibition within the network, at work with solely one map $\boldsymbol \eta^1$, reads as
\begin{align}\label{H1patt}
   H_N(\boldsymbol s | \boldsymbol \eta) = -\frac{1}{2N}\sum_{i,j=1}^{N,N} (\eta^1_i\cdot \eta^1_j) s_i s_j + \frac{\lambda-1}{2N}\sum_{i,j=1}^{N,N} s_i s_j = -\frac{N}{2} (x_1)^2 + (\lambda-1)\frac{N}{2} m^2,
\end{align}
where, in the right-most side, we expressed the Hamiltonian in terms of the order parameters $x_1$ and $m$ (as defined in eqs.~\eqref{eq:one_order}-\eqref{eq:two_order}).
\end{definition}
\begin{definition}(Guerra Action)
Following the mechanical analogy, we introduce a fictitious time $t \in \mathbb{R}^+$ and a ($1+2$) fictitious space $(y,\vec{z})\in \mathbb{R}\times\mathbb{R}^2$ that we use to define the interpolating free energy (or {\em Guerra action} \cite{Fachechi1}) $A(t, y, \vec z)$ as
\begin{align}\label{AinterpolataHJ}
    A(t,y,\vec{z}) = \frac{1}{N} \mathbb E_\eta \slr{\ln \sum_{s} \exp\lr{t \lr{\frac{1}{2} \sum_{ij}^{N,N} \eta_i^1\: \eta_j^1 s_i s_j -\frac{\lambda-1}{2} \sum_{ij}^{N,N} s_i s_j} + y \sqrt{\lambda-1}\sum_i^N s_i + z\cdot \sum_i^N \eta_i s_i } }.
\end{align}
\end{definition}
\begin{remark}
Note that, by setting $t=\beta$ and $(y, \vec z)=(0,\vec 0)$ in the above Guerra action, the latter coincides with the original free-energy \eqref{eq:Adefinition} of the Battaglia-Treves model in the low-storage limit.
\newline
Further, we observe that the {\em time} is coupled to the two-body contributions (the {\em energy} in the mechanical analogy), while {\em space} is coupled to the one-body contributions (the {\em momentum} in the mechanical analogy) as naively expected from general principles of Analytical Mechanics. 
\end{remark}
\begin{proposition}(Hamilton-Jacobi PDE)\label{PrepuzzioUno}
The interpolating free energy \eqref{AinterpolataHJ}, related to the model described by the Hamiltonian \eqref{H1patt}, obeys by construction the following Hamilton-Jacobi equation in the $3+1$ space-time $(y,\vec{z},t)$
\begin{subequations}
\begin{align}[left = {\empheqlbrace}]
   &\pder[A(t,y,\vec{z})]{t}  + \frac{1}{2} [\nabla A(t,y,\vec{z})]^2 + \mathcal V(t;y,z) = 0 \label{eq:HJ}\\
   &\mathcal V(t,y,\vec{z}) = \frac{1}{2} \slr{ \avg{x^2_1} -  \avg{x_1}^2} + \frac{1}{2} \slr{\avg{m^2} -\avg{m}^2}\label{eq:HJ_V}
\end{align}
\end{subequations}
where the gradient reads as $\nabla A(t,y,\vec{z}) \equiv \lr{\pder[A(t,y,\vec{z})]{y},\pder[A(t,y,\vec{z})]{\vec z}}$.
\end{proposition}
\begin{proof}
The proof works by direct evaluation of the derivatives (w.r.t. $y, \vec z$ and $t$) of the Guerra's action \eqref{AinterpolataHJ}.    
\newline   
The $t-$derivative of Guerra action $A(t;y,z)$, can be obtained straightforwardly as
\begin{align} \label{eq:pr1}
    \pder[A(t;y,z)]{t} = \frac{1}{2} \avg{x_1^2} + \frac{1-\lambda}{2} \avg{m^2},
\end{align}
while its gradient is
\begin{align}\label{eq:pr2}
    \nabla A(t;y,z) \equiv \lr{\pder[A(t;y,\vec z)]{y},\pder[A(t;y,\vec z)]{z}} = \lr{\avg{x_1}, \sqrt{1-\lambda} \:\avg{m}}.
\end{align}
\end{proof}
\begin{remark}
Note that the Hamilton-Jacobi equation can also be written as 
$$
\frac{\partial A(t,y,\vec{z})}{\partial t} + \mathcal{H}(t,y,\vec{z}) = 0,
$$
where the effective Hamiltonian $\mathcal{H}(t,y,\vec{z}) \equiv \mathcal T(t,y,\vec{z}) + \mathcal V(t,y,\vec{z})$ is comprised of a kinetic term $\mathcal T(t,y,\vec{z}) = \frac{1}{2} [\nabla A(t,y,\vec{z})]^2$ and of a potential term $\mathcal V(t,y,\vec{z}) = \frac{1}{2} \slr{ \avg{x^2_1} -  \avg{x_1}^2} + \frac{1}{2} \slr{\avg{m^2} -\avg{m}^2}$.
\end{remark}
\begin{remark}
Note that, as the potential $\mathcal V(t, y, \vec z)$ is the sum of two variances -- namely the two variances of two order parameters that are self-averaging as $N\to \infty$, see Definition \ref{selfaverage} -- we have that $\lim_{N \to \infty}\mathcal V(t,y,\vec{z}) =0$ hence, in this asymptotic regime of interest in Statistical Mechanics, the above PDE describes a free motion whose trajectories are Galilean.
\newline
For the sake of consistency with the brackets, clearly $\overline x \equiv \langle x \rangle$ and $\overline m \equiv \langle m \rangle$.
\end{remark}
To evaluate the explicit form of the Guerra action we need the following
\begin{theorem}\label{HJ-theorem}
The solution of the Hamilton-Jacobi PDE can be written as the Cauchy condition $A(t=0,y,\vec{z})$ plus the integral of the Lagrangian $\mathcal L (t,y,\vec z) = \mathcal T(t,y,\vec z) - \mathcal V(t,y,\vec z)$ over time, namely
\begin{align} \label{eq:AL}
    A(t;y,\vec z) = A(t=0;y_0,\vec z_0) + \int_0^t dt' \: \mathcal L (t,y,\vec z).
\end{align}
Explicitly and in the thermodynamic limit the solution reads as
\begin{align}\label{HJsolution}
    A(t,y,\vec{z}) = \mathbb E_\eta \ln \lr{1+\exp\lr{(y(t)-v_y t) \sqrt{1-\lambda} + (z(t)-v_z t) \eta}} + \frac{1-\lambda}{2} \avg{m^2} t + \frac{1}{2} \avg{y_1^2} t.
\end{align}
\end{theorem}
\begin{corollary}
The solution of the statistical mechanical problem, namely the expression for the free energy of the Battaglia-Treves model in the low storage regime, can be obtained simply by setting $t=-\beta$ and $y=\vec z=0$ in the expression for $A(t;y,\vec z)$ provided in eq. \eqref{HJsolution}, namely
\begin{align}
  A(\alpha=0,\beta,\lambda) \equiv  A(t=\beta;y=0,\vec z = \vec 0) = \mathbb E_\eta \ln \lr{1+\exp\lr{\beta\lr{1-\lambda} \avg{m} + \beta \eta \avg{x_1}}} - \beta \frac{1-\lambda}{2} \avg{m^2} - \frac{\beta}{2} \avg{x_1^2}.
\end{align}
It is a straightforward exercise to prove that, in order to extremize the above free energy w.r.t. the order parameters, their mean values $\langle x_1 \rangle$ and $\langle m \rangle$ have to obey the self-consistency equations \eqref{SC-Low2}-\eqref{SC-Low1}.
\end{corollary}
\begin{proof}
To solve for the problem posed in Proposition \ref{PrepuzzioUno}, we are left with two calculations to perform: the evaluation of the Cauchy condition $A(t=0;y_0,\vec z_0)$ and the integral of the Lagrangian over time.
\newline
The evaluation of the Cauchy condition is straightforward since, at $t=0$, neurons do not interact (see eq. \eqref{AinterpolataHJ}) and we get
\begin{align} \label{eq:A0}
    A(t=0;y_0,\vec z_0) = \frac{1}{N} \mathbb E_\eta \ln \prod_i \sum_{s_i=0,1} \exp\lr{y_0 \sqrt{1-\lambda}\: s_i + \vec z_0 \eta_i s_i} = \mathbb E_\eta \ln \lr{1+\exp\lr{y_0 \sqrt{1-\lambda} + \vec z_0 \eta}}.
\end{align}
Evaluating the integral of the Lagrangian  over time is trivial too (as it is just the multiplication of the kinetic energy times the time as the potential is null in the thermodynamic limit) and returns
\begin{align} \label{eq:L}
    \int_0^t dt'\:\mathcal L = \frac{1-\lambda}{2} \avg{m^2} t + \frac{1}{2} \avg{x_1^2} t.
\end{align}
Hence, plugging \eqref{eq:A0} and \eqref{eq:L} into \eqref{eq:AL}, we get
\begin{align}
    A(t;y,\vec z) = \mathbb E_\eta \ln \lr{1+\exp\lr{(y(t)-v_y t) \sqrt{1-\lambda} + (\vec z(t)-\vec v_z t) \eta}} + \frac{1-\lambda}{2} \avg{m^2} t + \frac{1}{2} \avg{x_1^2} t,
\end{align}
where we used $y_0 = y(t)-v_y t$ and $\vec z_0 = \vec z(t) - \vec v_z t$ as the trajectories are Galilean. Also, the velocities are
\begin{align}
    &v_y = \der[A(t; y,z)]{y} = \sqrt{1-\lambda} \: \avg{m},\\
    &\vec v_z = \der[A(t;y,z)]{\vec z} = \avg{x_1}.
\end{align}
Finally, the original free energy is recovered by setting $t=\beta$ and $y=\vec{z}=0$, namely
\begin{align}
  A(\alpha=0,\beta,\lambda) \equiv  A(-\beta;0,0) = \mathbb E_\eta \ln \lr{1+\exp\lr{\beta\lr{1-\lambda} \avg{m} + \beta \eta \avg{x_1}}} - \beta \frac{1-\lambda}{2} \avg{m^2} - \frac{\beta}{2} \avg{x_1^2}.
\end{align}
As expected, this expression coincides with eq. \eqref{eq:ARS} in the $\alpha\to 0$ limit and, its extremization w.r.t. the order parameters returns the two self-consistency equations \eqref{SC-Low2} and \eqref{SC-Low1}.
\end{proof}

\subsection{Phase diagram in the map high-storage regime}
\label{Guerra}

In the high-load regime ($\alpha >0$), we can not rely on the mechanical analogy as the concentration of measure argument used to kill the potential in the Hamilton-Jacobi PDE no longer works\footnote{Indeed, while each single variance could still be negligible (vanishing at a rate $1/N$ in the thermodynamic limit), now we are summing over an extensive number of them (as $K$ grows linearly with $N$ in the high-storage setting, that is for $\alpha >0$).}. Thus, here, we exploit the classical one-parameter interpolation based on stochastic stability \cite{guerra_broken}, adapted to deal with neural networks, see e.g. \cite{GuerraNN,Fachechi1}.

As standard in the high-storage investigation, we assume that a finite number of charts (actually just one) is retrieved and this map, say $\mu=1$, plays as the signal, while the remaining ones (i.e., those with labels $\nu \neq 1$) play as quenched noise against the formation of the coherent state for the retrieval of $\eta^{1}$.

The idea to solve for the free energy in this setting is to introduce an interpolating parameter $t\in[0,1]$ and an interpolating free energy $\mathcal A(t)$ such that, when $t=1$, the interpolating free energy recovers the free energy of the original model, i.e., $\mathcal A(t=1) = A(\alpha, \beta, \lambda)$, while, when $t=0$, the interpolating free energy recovers the free energy of an ``easy'' one-body system (where neurons interact with a suitably-constructed external field but their activity is no longer affected by the state of the other neurons). The main theorem we use in this section is the fundamental theorem of calculus  that plays as the natural bridge between these two extrema, as  
\begin{eqnarray}\label{InterMilan}
    A(\alpha, \beta, \lambda)&=&\mathcal A(t=1) = \mathcal A(t=0) + \int_0^1 ds \left. \left[\frac{d}{dt} \mathcal A(t)\right] \right \vert_{t=s}\\ \label{Mortaccitua}
    \mathcal A(t) &=& \lim_{N\to\infty} \frac{1}{N} \mathbb E_\eta  \ln \mathcal Z(t),
 \end{eqnarray}
where $\mathcal Z(t)$ is the interpolating partition function defined hereafter
\begin{definition}(Interpolating partition function)
The interpolating partition function $\mathcal Z(t)$ for the binary Battaglia-Treves model in the high-storage can be expressed as
\begin{equation}\label{eq:int_Z}
    \mathcal Z(t) = \sum_{s}^{2^N} \int d^d\mu(z_{\mu})\:\exp\lr{-\beta \mathcal H(t)},
\end{equation}  
where $\mathcal H(t)$ is the interpolation Hamiltonian and the Gaussian measure $d^d\mu(z_{\mu}) = \lr{\prod_{\mu>1} \frac{d^d z_\mu}{\sqrt{2\pi}} \exp\lr{-\frac{z_\mu^2}{2}}}$ has been introduced after applying the Hubbard-Stratonovich transformation to the Boltzmann-Gibbs measure of the model in order to linearize the quadratic contributions of the quenched noise (\emph{i.e.}, all the terms  $\mu>1$ in the Hamiltonian \eqref{eq:H_Hop}), resulting in the introduction of the $d-$dimensional \emph{i.i.d.} Gaussian distributed hidden neurons $z_\mu$ as follows:
\begin{align}\label{AgainDue}
    \exp\lr{\frac{1}{2} \frac{\beta}{N} \lr{\sum_{i,\mu>1} \eta^\mu_i s_i}^2} = \int d^d\mu(x) \:\exp\lr{\sqrt{\frac{\beta}{N}} \sum_{i,\mu>1} \eta^\mu_i \cdot z_\mu s_i}.
\end{align}
\end{definition}
The interpolating free energy is in turn based on the interpolating Hamiltonian $\mathcal H(t)$: the latter is a functional that combines two Hamiltonians, the \emph{true} Hamiltonian \eqref{eq:H_Hop} (obtained by setting $t=1$) and a \emph{new} Hamiltonian, $\mathcal H_0$ (obtained by setting $t=0$). The new Hamiltonian has to be constructed {\em ad hoc} with two requisites: it should allow an analytical treatment of its free energy (typically it has to be a sum of one-body models) and the effective field that it produces on the neuron $s_i$ has to mimic the true post synaptic potential  (and this is achieved by using linear combinations of independent random variables).   
\begin{definition}(Interpolating Hamiltonian)\label{InterpoAcca}  
Given an interpolating parameter $t \in [0,1]$ and the real-valued function $\phi(t)$, the interpolation Hamiltonian $\mathcal H(t)$ reads as
\begin{eqnarray}\label{eq:int_H}
-\beta \mathcal H(t) &=& \frac{t}{2} N\beta x_1^2 + \frac{t}{2} N \beta \lr{1-\lambda} m^2 + \sqrt{t}\sqrt{ \frac{\beta}{N}} \sum_{i,\mu>1} \eta^\mu_i s_i z_\mu + N \phi(t),\\ \label{One-Body-Term}
N\phi(t) &=& \phi_1(t)\sum_{\mu>1} x_\mu^2 + \phi_2(t) \sum_\mu \rho_\mu x_\mu + \phi_3(t) \sum_i h_i s_i + \phi_4(t) \sum_i \eta^1_i s_i + \phi_5(t) \sum_i s_i^2 + \phi_6(t) \sum_i s_i.
\end{eqnarray} 
The auxiliary functions $\phi_i(t)$ must obey $\phi_i(t=1) = 0$ (such that   $\phi(t=1) = 0$) and, when $t=0$, $\mathcal H(t=0)$ must be built of by effective one-body contributions only.
The specific functional form of $\phi(t)$, as well as its derivation, are provided in Appendix~\ref{GuerraAPP}. 
\end{definition}
\begin{remark}
Note that, in the integral representation of the partition function achieved by the Hubbard-Stratonovich transformation (see \eqref{AgainDue}), $K$ hidden variables  $z_{\mu}$  (selectively firing one per chart) naturally arise within the theory and, as a consequence, their related overlaps must be introduced as auxiliary order parameters, namely
\begin{eqnarray}
p_{12} &=& \frac{1}{K}\sum_{\mu=1}^K z_{\mu}^1  z_{\mu}^2,\\   
p_{11} &=& \frac{1}{K}\sum_{\mu=1}^K z_{\mu}^1  z_{\mu}^1.
\end{eqnarray}
However, these overlaps do not deserve a dedicated definition as they are dummy variables that will disappear in the final expression for the free energy.
\end{remark}
\begin{definition}(Generalized average) The generalized average $\langle \cdot\rangle_t$ is defined as
$$
\langle \cdot \rangle_t = \mathbb E_\eta [\omega_t(\cdot )],
$$    
where $\omega_t(\cdot )$ is the Boltzmann average stemming from the interpolating Boltzmann factor $\exp\lr{-\beta \mathcal H(t)}$ defined in eq. \eqref{eq:int_H} together with the interpolating partition function $\mathcal Z(t)$ defined in eq. \eqref{eq:int_Z}, while the operator $\mathbb E_\eta$ still averages over the quenched maps.
\end{definition}
In the following, if not otherwise specified, we refer to the generalized averages simply as $\langle \cdot \rangle$ in order to lighten the notation.\\
By a glance at \eqref{InterMilan} we see that we have to evaluate the Cauchy condition  $\mathcal A(t=0)$ -- that is straightforward as in $t=0$ the Gibbs measure is factorized over the neural activities -- and integrate its derivative  $\frac{d\mathcal A(t)}{dt}$ -- that is more cumbersome. 
\newline
Starting from the evaluation of the $t$-derivative of the interpolating free energy  $\mathcal A(t)$, we state the next 
\begin{proposition}\label{StreamingHSR}
Under the RS assumption and in the thermodynamic limit, the $t-$derivative of $\mathcal A$ can be written as
\begin{align}
    \frac{d\mathcal A_{RS}}{dt} = - \frac{\beta}{2} (1-\lambda)\overline m^2 - \frac{\beta}{2} \overline x^2 - \frac{\alpha \beta}{2d} \lr{\overline p_1 \overline q_1 - \overline p_2 \overline q_2}. 
\end{align}    
\end{proposition}
\begin{proof}
The proof works by direct calculation and by requiring that $\lim_{N \to \infty}\mathcal{P}(q_{11})=\delta \left(q_{11}-\bar{q}_1\right)$ and the same for $q_{12}, \ p_{11},\ p_{12}$.  See Appendix~\ref{GuerraAPP} for details.   
\end{proof}
For the Cauchy condition we can state the next
\begin{proposition}\label{CauchyHS}
The Cauchy condition $\mathcal A(t=0)$, related to the expression of the interpolating free energy $\mathcal A(t)$ provided in eq. \eqref{Mortaccitua}, in the thermodynamic limit reads as
\begin{align}
    \mathcal A(t=0) = \frac{\alpha d}{2}\ln \frac{2\pi}{1-2\phi_1(0)} + \frac{\alpha}{2} \mathbb E_\rho\lr{ \frac{\phi^2_2(0) \rho^2}{1-2\phi_1(0)} }+ \mathbb E_{\eta^1} \int Dz \ln \lr{1+ \exp\lr{\phi_3(0) z + \phi_4(0) \eta_i^1 + \phi_5(0) + \phi_6(0)}},
\end{align}
where we have rewritten $h_i$ as the Gaussian variable $z \in \mathcal N(0,1)$.
\end{proposition}
\begin{proof}
The proof is a direct trivial calculation (as at $t=0$ there are no interactions among neural activities) with the usage of the self-averaging assumption of the order parameters. See Appendix~\ref{GuerraAPP} for details.
\end{proof}
These results merge together thanks to the fundamental theorem of calculus (see \eqref{InterMilan}) that we use to state the next  
\begin{theorem}\label{TzeTze}
In the thermodynamic limit, the replica-symmetric quenched free-energy of the Battaglia-Treves
model, equipped with McCulloch-Pitts neurons as defined in eq. \eqref{def:cost_function}, for the $S^{d-1}$ embedding space, can be expressed in terms of the (mean values of the) order parameters $\overline{m}, \overline{x}, \overline{q}_1, \overline{q}_2, \overline{p}_1, \overline{p}_2$ and of the control parameters $\alpha, \beta, \lambda$, as follows
\begin{align}
A_{RS} =& - \frac{\beta}{2} (1-\lambda)\overline m^2 - \frac{\beta}{2} \overline x^2 - \frac{\alpha \beta}{2d} \lr{\overline p_1 \overline q_1 - \overline p_2 \overline q_2} - \frac{\alpha d}{2} \ln \lr{1-\frac{\beta}{d}(\overline q_1-\overline q_2)} + \frac{\alpha \beta}{2}\frac{\overline q_2}{1-\frac{\beta}{d}(\overline q_1-\overline q_2)} +\nonumber\\
&+ \mathbb E_{\eta} \int Dz \ln \lr{1+ \exp\lr{ \beta (1-\lambda) \overline m + \beta \overline x \cdot \eta + \frac{\alpha\beta}{2d}(\overline p_1-\overline p_2) + \sqrt{\frac{\alpha \beta}{d} \overline p_2} \:z }}. \label{TeoExt}
\end{align}
\newline
Chart cell's overlaps $p_1$ and $p_2$ can be substituted by their saddle-point values, obtained by differentiating $A_{RS}$ w.r.t. $q_1$ and $q_2$, \emph{i.e.} $\pder[A^{RS}]{q_1} = 0, \pder[A^{RS}]{q_2} = 0$, 
\begin{align}\label{GDC-O1}
    &\overline p_2 = \frac{\beta \overline q_2}{\lr{1-\frac{\beta}{d}(\overline q_1-\overline q_2)}^2},\\ \label{GDC-O2}
    &\overline p_1-\overline p_2 = \frac{d}{1-\frac{\beta}{d}(\overline q_1-\overline q_2)}.
\end{align}
\end{theorem}
\begin{proof}
See Appendix~\ref{GuerraAPP}.     
\end{proof}
\begin{corollary}
By inserting the expressions for the chart cell overlaps reported at the r.h.s. of eq.s \eqref{GDC-O1}, \eqref{GDC-O2} within the expression of the  replica-symmetric quenched free-energy \eqref{TeoExt} we obtain the r.h.s. of eq. \eqref{eq:ARS} and thus Theorem $1$ is recovered. 
\end{corollary}
\begin{remark}\label{Appendici}
For the sake of completeness we also provide the same analysis, performed via the standard replica-trick rather than with Guerra's interpolation, in Appendix \ref{Replicas}. Furthermore, by looking for the extremal points of such a free energy, it is a simple exercise to obtain  the same self-consistencies for the order parameters reported in eq.s \eqref{eq:tre}-\eqref{eq:due}.
\end{remark}

\section{From rats on a circular track to bats in the tunnel}\label{sec:Tunnel}
The model so far investigated was originally developed to mimic real neurons in rats exploring (or rather foraging in) small boxes, or short tracks, each of which could be idealized as being represented by an unrelated chart, in which each neuron has a place field of standard size. We can now turn it into a model inspired by recent experiments with bats \cite{Batman1}, of the neural representation of a single extended environment, expressed by place fields of widely different sizes, from very long to very short. For simplicity, we stick again to 1D environments.

Indeed, recent experiments recording CA1 place cells either in rats running on long tracks \cite{fenton2008unmasking,rich2014large} or bats flying in long tunnels \cite{Batman1} have evidenced that individual place cells can have multiple place fields, of remarkable variability in size and peak firing rate. In the following, extending the formalism developed above, we consider a mathematical model which assigns this variability to distinct portions of the environment, representing them with a sequence of gross, or fine-grained, charts. Alongside "traditional" place cells, which participate in every chart, and hence have multiple place fields of variable width, the model envisages "chart" cells, which are active only when the current position is within their chart. Note that the very broad distribution of place field widths observed in particular in bats (where some place fields are smaller than a meter, while others are larger than 30 meters, resulting in a log-normal distribution), is not reproduced at the individual cell level, in the model, but only as a population-coherent scaling of the local metric.  
\newline
This can be seen as a variant of the Battaglia-Treves model, motivated by mathematics more than by neuroscience. 
\newline
Instead of covering the whole space with a single map, or chart, we combine a number $K=\alpha N$ of maps $\eta^\mu$ to split the embedding space $I_L=[0,L]$ (the one-dimensional interval that represents the environment) in patches, by distributing the place field's centers at different locations $\overline r^\mu \in I_L$ and assigning them different sizes, denoted by the parameters $\sigma^\mu$ ({\em vide infra}).
\newline
The characteristic width of the coherent states that are reconstructed by the model (previously indicated by $\delta$), beyond being field-dependent now,  has also to be kept much smaller than the length of the tunnel, \emph{i.e.} $\delta\ll L$ in order to reconstruct the position of the animal with good accuracy. 
\newline
Secondly, driven by the duality (between the original place cells, now with multiple fields along $L$, and the chart cells, acting as hidden units), we show that by suitably coupling these hidden variables we can turn the original model for chart reconstruction into a behavioral model, namely into a model for spatial navigation too, both externally (see Sec. \ref{external-driven}) or internally  (see Sec. \ref{internal-driven}) driven. 
\newline
Indeed,  the environment sensed by the animal is now dynamical (i.e. it is no longer the same constant metric perceived when confined in a small cage) hence, in Sec. \ref{external-driven},  we  introduce a time-dependent external field (i.e. the input perceived by the bat), that is moving at its same speed  and guides the animal trough the tunnel: see Figure \ref{fig:tcycle}. 
\newline   
In Sec. \ref{internal-driven}, instead, we assume that the bat has learnt how place fields are dislocated within the environment and we show how, by introducing a simple coupling among their corresponding chart cells, such a minimal generalization of the Battaglia-Treves model can account for autonomous motion within the environment.

\begin{figure}[!h]
    \centering
    \includegraphics[width=18cm]{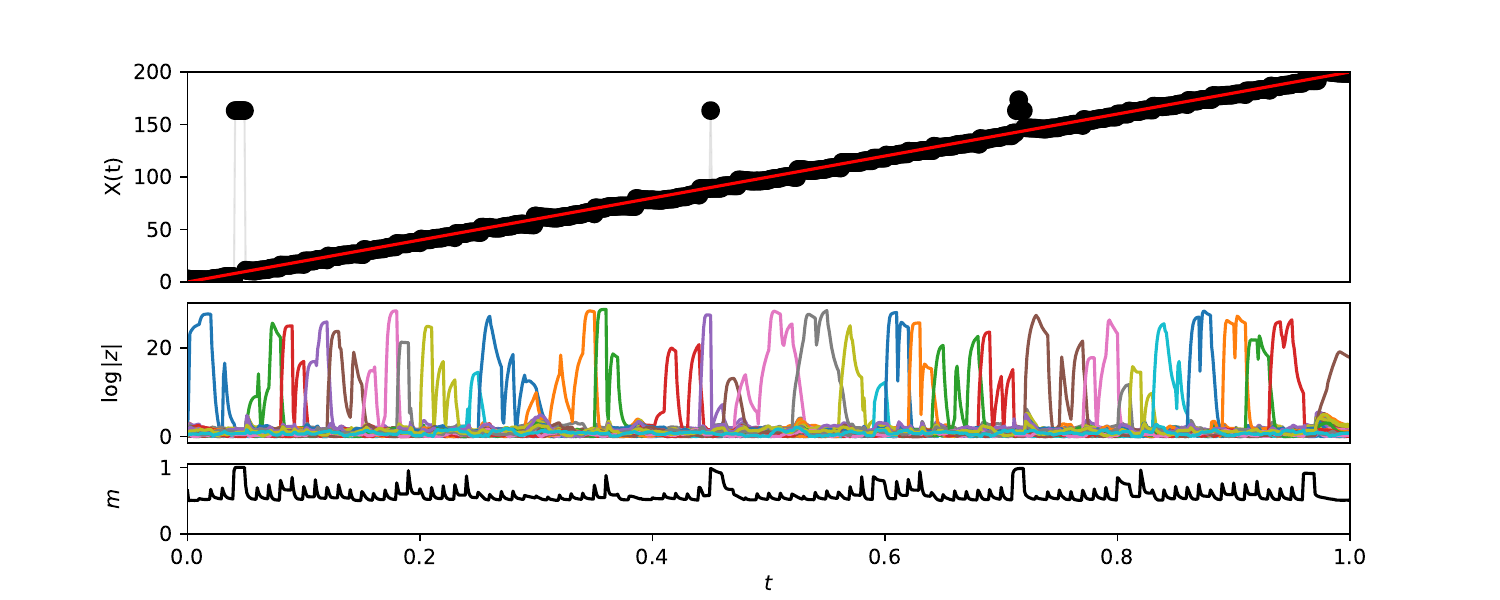}
    \caption{Simulation of the  bat traveling in the tunnel driven by the external field $h(t)$. \emph{Top} The motion of the bat follows the field, as expected, almost everywhere apart from a small number of points where the model fails to reconstruct the position (i.e. the upper black points away from the red straight line). \emph{Middle} Chart cell firing patterns during the motion of the animal in the tunnel are highly correlated with the position of the animal. \emph{Bottom} The overall activity of the model is also shown as a function of time and  the points where chart reconstruction fails correspond to higher level of neural activity. \\
    The simulation features $N=8000$ visible neurons $\mathbf s$ and $K=40$ chart cells $\mathbf z$, at a temperature of $\beta^{-1}=10^{-3}$.}
    \label{fig:tcycle}
\end{figure}

\begin{figure}
    \centering
    \includegraphics[width=10cm]{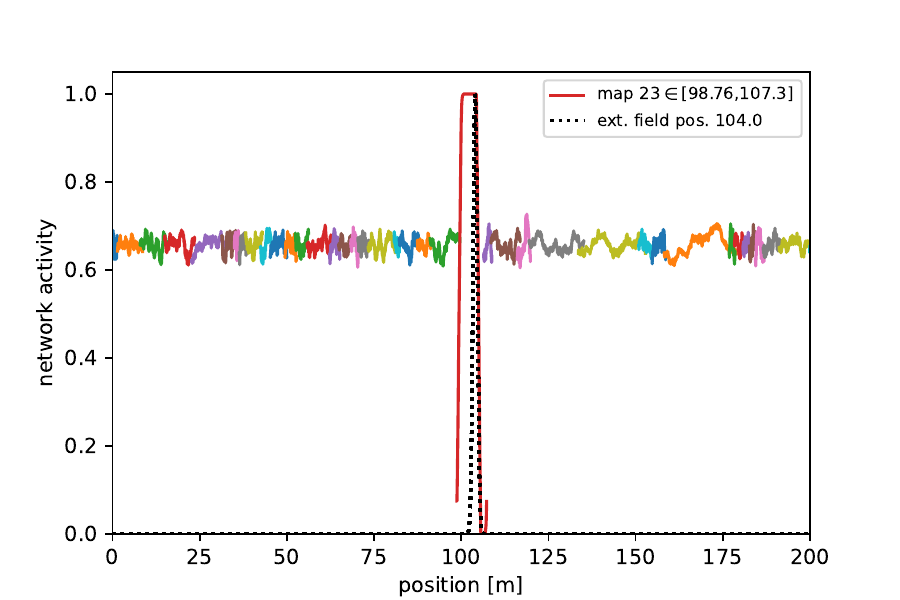}
    \caption{A simulation of the reconstruction process that allows the network to recognize the position of the external field $\mathbf{h}$ (in this case $104.0$m). The model activity is shown for each chart (for a total of $K=40$ charts) distributed along the tunnel at different positions: the chart with the highest activity is the number $23^{th}$ (with an overlap of $\overline x_{23}=0.28$ in this case) which covers the tunnel's space in the interval $[98.76,107.3]$m (as shown in the legend): in that space  the external Gaussian bump is actually found.}
    \label{fig:rec}
\end{figure}

\subsection{Numerical experiments part one: external-driven motion}\label{external-driven}
We assume that the space along the tunnel is uniformly covered with maps, but each maps has a different width $\sigma^\mu$, with distribution $\rho\lr{\sigma^\mu}$. According with experiments \cite{Batman1}, we assume that $\rho\lr{\sigma^\mu}$ is \emph{log-normal}, while the density of neurons in each chart, $\rho\lr{r^\mu_i|\overline r^\mu;\sigma^\mu}$, is uniform. These assumptions are here summarized:
\begin{align}
    &\rho\lr{\sigma^\mu|\overline r^\mu, \sigma} = \log \mathcal N(\overline r^\mu,\sigma),\label{eq:rhos0}\\   \label{eq:rhos1}
    &\rho\lr{r^\mu_i|\overline r^\mu;\sigma^\mu}=\mathcal U_{I^\mu},
\end{align}
where $I^\mu = \left[\overline r^\mu-\frac{\sigma^\mu}{2},\overline r^\mu+\frac{\sigma^\mu}{2}\right]$ is the interval centered around $\overline r^\mu$ and width $\sigma^\mu$; the parameter  $\sigma$  appearing in eq. \eqref{eq:rhos0} is a free parameter  of the model.
Note that each map is periodic in its domain, $I^\mu$, given the definition of the map coordinates
\begin{align}
    \eta^\mu(r|\overline r^\mu,\sigma^{\mu}) = \lr{\cos\lr{\frac{2\pi}{\sigma^\mu} (r-\overline r^\mu)},\sin\lr{\frac{2\pi}{\sigma^\mu} (r-\overline r^\mu)}},
\end{align}
where the chart centers $\overline r^\mu$ are distributed according to a prior $\mathcal{P}(\overline r^\mu)$. 
Herafter,  we call $t$ the time  such that the sensory input of the animal in the tunnel is represented by a time dependent external field $h(t)$, which is added to the bare Hamiltonian $H$ of the model as $H(t)=H + \sum_i h_i(t) s_i$. This new term in the energy, i.e. $\sum_i h_i(t) s_i$ account for the motion of the bat and it produces a bias in the update equations for the neural dynamics (see the stochastic process coded in eq. \eqref{dinamical}), which can be explicitly written as
\begin{align}\label{eq:dyn}
    &s^{t+1}_i = \Theta\lr{\sigma\lr{\beta \mathcal V_i(\mathbf s^t)} - u_i},\:\:u_i\in \mathcal U_{[0,1]},\\
    &\mathcal V_i(\mathbf s^t) = \frac{1}{2}\sum_j J_{ij} s^t_j + h_i(t),
\end{align}
where $u_i\in \mathcal U_{[0,1]}$ is a random number uniformly distributed in the interval $[0,1]$, $\Theta(x)$ is the Heaviside step function and $\mathcal V_i(\mathbf s^t)$ is the overall post-synaptic potential (which depends on the synaptic matrix \eqref{eq:kernel1} and is comprehensive of the external field $h_i(t)$ too) acting on the $i$-th neuron $s_i$.\\
The explicitly time-dependent field, $h_i(t)$  enables the movement of the bump of neural activity representing the position of the bat along the tunnel if the external field follows the bat's motion. For the sake of simplicity, we assume its motion to be Galilean with constant velocity $v$, in such a way that $h(t)$ can be modeled as a Gaussian bump traveling with constant velocity $v$ and represented in the place cells coordinate system as
\begin{align*}
    h_i(t) = \exp\lr{-\frac{(r_i^\nu-vt)^2}{2\delta^2_0}},
\end{align*}
with $\delta^2_0$ the size of the bump; $r_i^\nu$ is the coordinate  relative to the  $\nu-$th grandmother cell, $\eta^\nu$; the index $\nu$ is chosen with the following criterion (see Figure \ref{fig:rec}): a collection of external fields $h^\mu(t)$ is computed for each map $\mu$, then the map with the highest average activity (that is, the map with the highest overlap with the external field) is chosen and the relative index is indicated as $\nu$\footnote{Note that the quantity $\frac{1}{N}\sum_i h_i(t)$ can be approximated, in the thermodynamic limit $N\to\infty$, by the integral 
$$\int_{\overline r^\nu-\frac{\sigma^\nu}{2}}^{\overline r^\nu+\frac{\sigma^\nu}{2}} dr\: h(r,t)\sim \sqrt{2\pi}\frac{\delta_0}{\sigma^\nu}$$ 
(as it is the integral on a finite domain of a Gaussian distribution whose center is $vt$ within the interval $[\overline r^\nu-\frac{\sigma^\nu}{2},\overline r^\nu+\frac{\sigma^\nu}{2}]$ and whose  standard deviation reads $\delta_0\ll \sigma^\nu$) that is an intensive quantity (as it only depends on the ratio of two intensive parameters $\delta_0,\sigma^\nu$), such that, overall, the external field contribution to the bare Hamiltonian $H$ is \emph{extensive} in $N$, as it should.}.

The velocity $v$ of the traveling bump introduces a new time scale in the dynamics of the system; this time scale, written as $\tau = L/(M_0 v)$, is the time that the bump has to spend when travelling in the tunnel $L$, divided by the number $M_0$ of discrete-time realizations of the field itself (hence, in the time window $M_0 \tau$, the Gaussian bump moves in a given number of steps $M_0$ from the origin of the tunnel  to its end). In order to allow the network to stabilize within a given fixed point of the dynamics \eqref{eq:dyn} at each new position at discrete time $t$, \emph{i.e.} $r_0 + vt$, we require the number of steps of the dynamics to be $M\gg M_0$. In our simulations we inspect in detail the case where $M_0 = 100$ and $M=1000$, with a neural noise level fixed at $\beta=100$: see Figures \ref{fig:rec} and \ref{fig:tcycle}. 

Under these assumptions, the dynamical update rules simplify to
\begin{align}
    &P(\mathbf z| \mathbf s) = \prod_\mu \mathcal N(\sqrt N x_\mu(\mathbf s), \mathbb I_d \:\beta^{-1}),\\
    &P(\mathbf s = 1 | \mathbf z) = \prod_i \sigma \lr{\frac{\beta}{\sqrt N} \sum_\mu \eta^\mu_i \cdot  z_\mu + \beta h_i}.
\end{align}
We run 1000 simulations of a bat flying in the tunnel (each one with a different realization of the chart's representation) and results are shown in Figure \ref{fig:tcycle}:  the network is able to successfully follow the external field $\mathbf{h}(t)$ and, as the bat flyies along the tunnel, all the various place cells sequentially fire. Furthermore, as expected, the  empirical histogram of the width of the fields related to the active chart cells is approximately log-Normal, as shown in fig. \ref{fig:tunnelhist}.
\begin{figure}[!h]
    \centering
    \includegraphics[width=10cm]{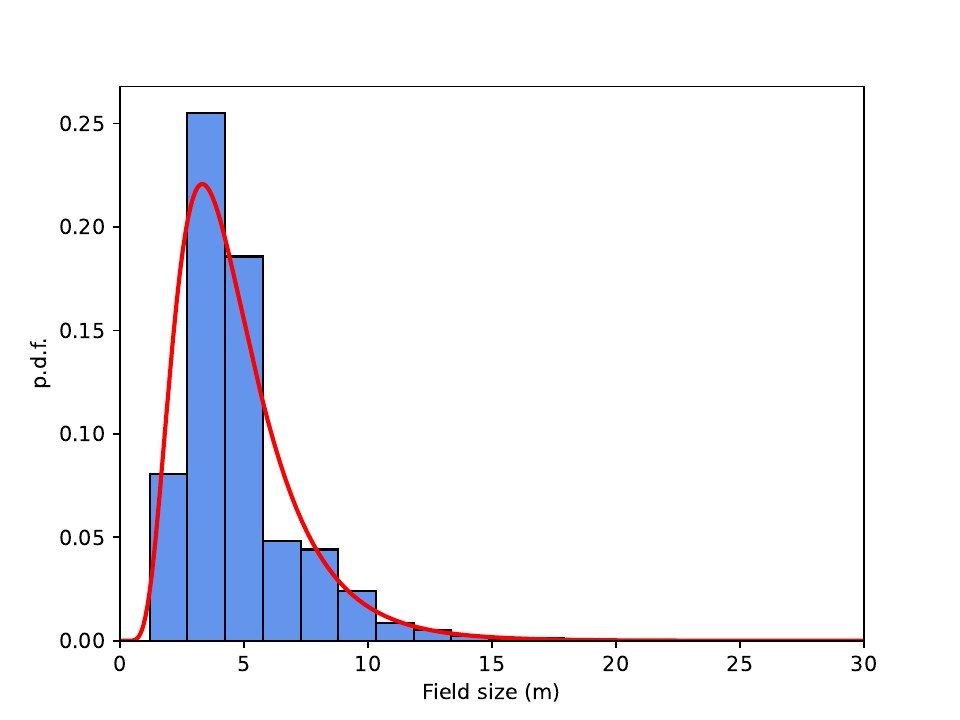}
    \caption{The histogram shows the result of $100$ simulated trajectories along the $200$m tunnel of the network, in terms of the width of the activation of the chart cells, which approximates a Log-Normal distribution with mean $e^{\mu_d}=4.2$m and $\sigma_d=0.484$m.\\
    Notice that this result is expected, since the chart cells were defined to map the tunnel with a similar Log-Normal distribution with $e^{\mu_0}=4.8$m and $\sigma_0=0.6$m.}
    \label{fig:tunnelhist}
\end{figure}

\subsection{Numerical experiments part one: self-driven motion}\label{internal-driven}
In the present setting, place cells reconstruct the environment they navigated and activate the corresponding chart cells as the animal goes through their several place fields. While simple place cells thus remap the familiar environment, they can not easily drive the movement of the animal within the environment \cite{PC-review1}  (simply because of the multiplicity of their place fields). However, a coupling between consecutive chart  cells -say $\mu$ and $\mu+1$- naturally accounts for a minimal model of bat's dynamics: this coupling could be easily implemented by adding it to the Hamiltonian of the original Battaglia-Treves model.    
\newline
As we distributed the charts sequentially in the tunnel, in the simulations, such that the bat encounters them one after the other (i.e., $\mu=1 \to \mu=2 \to ,... \to \mu=K$), we can extend the Battaglia-Treves Hamiltonian by introducing a coupling between chart cells (i.e., $J_z \sum_{\mu}^K z_{\mu}z_{\mu+1}$) as:
\begin{equation}\label{BTG}
H_{N,K}(s,z|\boldsymbol{\eta})= -\frac{1}{\sqrt{N}}\sum_i^N \sum_{\mu}^K \eta_i^{\mu} s_i z_{\mu} + \frac{(\lambda-1)}{2N}\sum_{i,j}^{N,N}s_i s_j  + J_z \sum_{\mu}^K z_{\mu}z_{\mu+1},    
\end{equation} 
in the standard hetero-associative way introduced by Amit \cite{AmitPnas} and Kosko \cite{Kosko}\footnote{See also \cite{BAM-1,BAM-1} for a recent re-visitation of the emergent computational capabilities of those networks analyzed with the techniques exploited in this paper.}.
\newline
We end up with a neural network that {\em drives} the bat correctly along the tunnel as the charts activated by the neurons $\boldsymbol{z}$ guides the bat, see Figure \ref{fig:tcycle3}, while the finer resolution place fields of neurons $\boldsymbol{s}$ represents on-line detailed locations crossed during the flight, see Figure \ref{fig:tcycle2}.
\begin{figure}[!h]
    \centering
    \includegraphics[width=15cm]{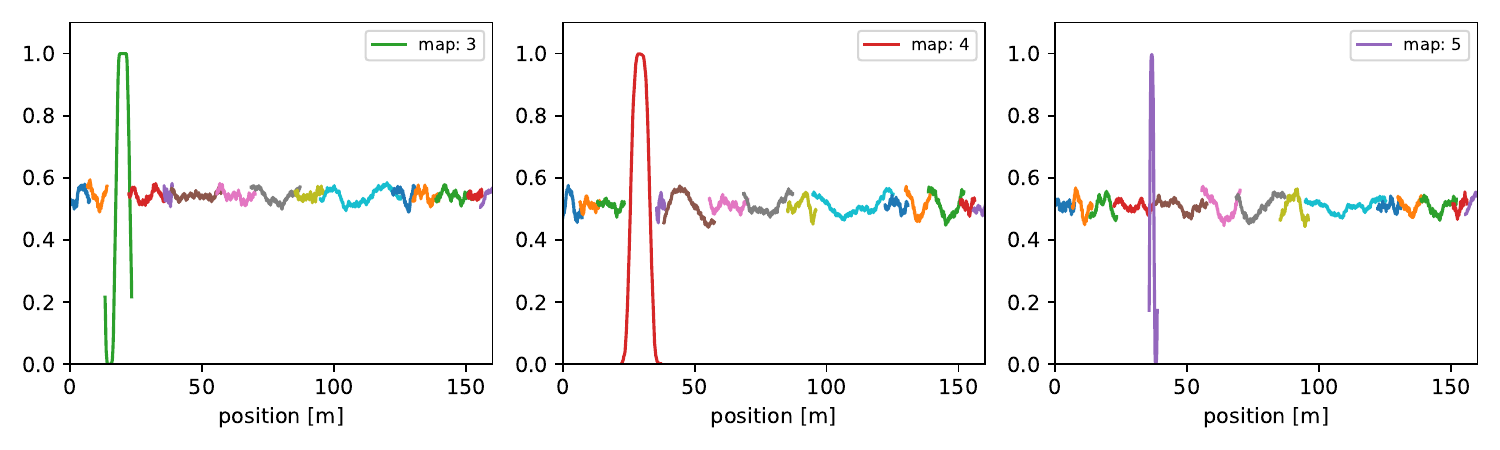}     
    \caption{Three consecutive snapshots of the activity of chart cells simulated during the dynamics  generated according to eq. \ref{BTG}. Note that, in the first panel, the bat retrieves map $3$ and thus the corresponding chart cell fires. This, in turn, drives the firing of the next chart cell, related to retrieval of map $4$ and so on.}
    \label{fig:tcycle3}
\end{figure}
\begin{figure}[!h]
    \centering
    \includegraphics[width=18cm]{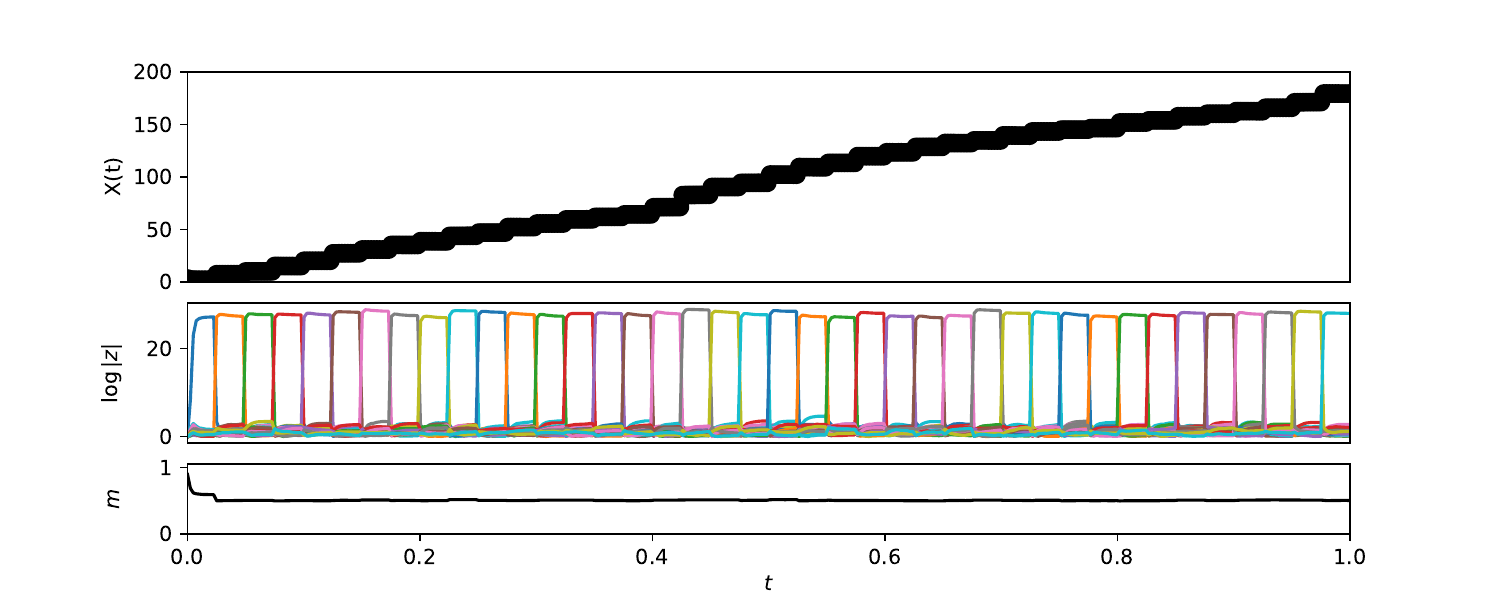}
    \caption{Simulation of the bat flying in the tunnel self-driven (i.e. driven by the coupling between chart cells $z_{\mu}z_{\mu+1}$ following the extension of the Battaglia-Treves model provided in eq. \ref{BTG}). \emph{Top} The position of the animal can now be thought of as moved along by its network and no external field $h(t)$ is required. \emph{Middle} Chart cells drive the motion of the animal by activating each other in a sequence; but -- note -- at a constant rate per cell, not at a constant bat speed. \emph{Bottom} The overall activity of the model remains around $m\sim 0.5$ as expected.\\
    The simulation features $N=8000$ visible neurons $\mathbf s$ and $K=40$ chart cells $\mathbf z$, at a temperature of $\beta^{-1}=10^{-3}$ and with $J_z \sim O(1)$.}
    \label{fig:tcycle2}
\end{figure} 
The update rule of the chart cell simply reads as
\begin{align}\label{eq:ztranslation}
    P(\mathbf z^{t+1}|\mathbf z^t, \mathbf s^t) = \prod_\mu \mathcal N(J_z z^t_{\mu+1} + \sqrt N \:\overline x^\mu(\mathbf s^t),\mathbb I_d \beta^{-1}).
\end{align}
Note that, in order for the rule \eqref{eq:ztranslation} to correctly reproduce the behavior of the bat, the timescale related to the $z$ variables must be slower w.r.t. the timescale related to the $s$ variables (such that the former may integrate information from the latter) in line with the common observation of timescales across different computational units (see e.g. \cite{Bernacchia1,Bernacchia2,Bernacchia3,Sandrone1}): in the simulations the ratio among these timescales has been set to $\sim O(10^{-1})$.

\bigskip
\section{Conclusions}\label{Conclusions}
Research on place cells, grid cells and, in general, the way hippocampus stores spatial representations of the environment is obviously a central theme in Neuroscience (see e.g. \cite{ThompsonRat,Solstad2006,Caswell2006,Batman1}).  Since the AGS milestone in the middle eighties \cite{AGS,Amit,amit1985storing,Coolen}, the {\em statistical mechanics of spin glasses} played a central role in the mathematical modeling of the collective, emergent features shown by large assemblies of neurons. hence it is not by surprise that computational and analytical investigations on place cells along statistical mechanics lines have been extensive (see e.g. \cite{Ale0,MonassonTreves2,MonassonPlaceCellsPRL}).
However, previous studies in this field based their findings on methodological tools (e.g. the so-called replica trick \cite{Amit,Coolen,Huang}) that, from a mathematical perspective, are somehow heuristic \cite{talagrand2003spin} thus raising the quest for confirmation by an independent, alternative, approach: in these regards, since the pioneering Guerra's works on spin glasses \cite{guerra_broken,HJ-Barra2010Aldo},  Guerra's interpolation techniques quickly became a  mathematical alternative to the replica trick also in neural networks (see e.g. \cite{GuerraNN,glassy,Fachechi1,Lenka1,Lenka2,Jean1,Jean2}) and these have been the underlying methodological leitmotif of the present paper too\footnote{For the sake of completeness, we point out that rigorous methods in the statistical mechanical formalization of neural networks are obviously not confined to Guerra interpolation and we just mention the books \cite{BovierReti,TalaReti} as classical milestones in this field, should the reader be interested.}. 
\newline
In particular, we studied analytically a variant of the Battaglia-Treves model \cite{battaglia1998attractor} for chart storage and reinstatement; a model that describes the collective emerging capabilities of place cells in rats exploring small boxes (where maps are encoded by uniform coverage of similarly sized place fields). Specifically, our variant of the network is equipped with McCulloch-Pitts neurons \cite{Coolen} and we study it both in the low-storage (where the number of stored charts scales sub-linearly w.r.t. the network size) and high-storage (where the number of stored charts scales linearly w.r.t. the network size) regimes \cite{Amit}. The former has been inspected by adapting the Hamilton-Jacobi (two-parameters) interpolation \cite{HJ-Barra2010Aldo,HJ-Barra2013Gino}, while the latter has been tackled by adapting the stochastic stability (one-parameter) interpolation \cite{GuerraNN,Fachechi1}. Further, we have also derived the same results independently via the replica trick in the high-storage limit (these calculations are reported in the Appendix  for the high storage regime but turn out to be coherent also with the low storage picture reported in the main text, simply by setting $\alpha=0$): confined to the replica symmetric level of description (fairly standard in neural networks \cite{Lenka}), as expected, we have obtained full agreement among the results stemming from  these nethodologies. Further, extensive Monte Carlo simulations for finite-size-scaling inspected the network behavior away from the thermodynamic limit and provided heuristic confirmation of the picture obtained under replica symmetry assumption. 
\newline
Interestingly, the integral representation of the partition function of the model that we use, within the Guerra approach, naturally highlights a dual representation of the Battaglia-Treves model: from a pure mathematical modelling perspective, this network has been shown to be equivalent to a bipartite network, equipped with a visible and a hidden layer of neurons. Neurons in the visible layer are those of the original model while neurons in the hidden layer have been shown to play as chart cells, selectively firing one per stored map or chart. 
\newline
Further, since simulations allow us to inspect models too cumbersome for a rigorous analytical treatment, we have numerically simulated an extension of the Battaglia-Treves model in order to model hippocampal cells in bats flying in large tunnels: the idea is in how charts are coded: now place fields are no longer uniformly distributed nor their width constant. Yet, by replacing the original uniform distributions for place field and chart width with the Log-Normal observed in recent experiments \cite{Batman1} (that prescribe an exponential distribution of place cells along the tunnel and a log-normal distribution of their width), we have shown that the Battaglia-Treves model succeeds in reproducing the reconstruction of a map by a bat flying within the tunnel.

\section*{Acknowledgments}
A.B. and M.S.C. acknowledge PRIN 2022 grant {\em Statistical Mechanics of Learning Machines: from algorithmic and information theoretical limits to new biologically inspired paradigms} n. 20229T9EAT funded by European Union - Next Generation EU. 
E.A. acknowledges financial support from PNRR MUR project PE0000013-FAIR and from Sapienza University of Rome (RM12117A8590B3FA, RM12218169691087).
A.B., E.A. are members of the  GNFM-INdAM which is acknowledged too as well as the INFN (Sezione di Lecce). A.B. is a member of INFN (Sezione di Lecce) that is also acknowledged.

\appendix 

\section{Free energy calculation via the Replica Trick}\label{Replicas}
As a check, in this appendix we provide a derivation for the free energy of the model also with the old fashioned replica trick, to guarantee that results do coincide with those obtained via the Guerra routes in the main text.
\newline
As the Battaglia-Treves model is ultimately a spin glass model, it is rather natural to tackle the evaluation of its free energy by the formula
\begin{equation}
   A(\alpha,\beta,\lambda)= \lim_{N\to \infty} \lim_{n \to 0}   \frac{Z_N(\beta,\lambda)-1}{nN}, 
\end{equation}
thus with the computational reward of bartering the evaluation of the logarithm of the partition function with its momenta: the price to pay for this reward is the \emph{blind} analytical extension toward the limit of zero replicas $n \to 0$ that must be performed (under the assumptions that the limits commute, i.e. $[\lim_{N\to \infty},\lim_{n \to 0}]=0$)\footnote{While this procedure has been proved to be correct for the harmonic oscillator of spin glasses, that is the Sherrington-Kirkpatrick model \cite{MPV}, by constraining the Parisi expression for its free energy among the Guerra \cite{guerra_broken} and the Talagrand \cite{TalaParisi} upper and lower bounds (and it is also true that, for that model,  $[\lim_{N\to \infty},\lim_{n \to 0}]=0$ 
 \cite{Mingione}), at present, in neural networks nor we do have a general prescription for broken replica theories a' la Parisi neither we can be sure (in the mathematical sense) that the formula we obtain by the replica trick are ultimately correct, each time we use it, hence the need for alternative approaches.}. 
\newline
As a consequence we are left to evaluate the moments of the partition function, often called {\em replicated partition function}, $Z^n$ that reads 
\begin{align}
    Z^n = \int d\mu(s^a) \exp \slr{\frac{\beta}{2N} \sum_{\mu,a} \lr{\sum_{i}\eta^\mu_i s^a_i}^2 - \frac{\beta}{2N} (\lambda -1) \sum_a \lr{\sum_{i} s^a_i}^2};
\end{align}
As standard we assume the network lies in the basin of attraction of one map, say $\eta^1_i$, hence we split the signal term (provided by that map) from the background noise (resulting from all the other maps $\nu \neq 1$) by rewriting the replicated partition function as
\begin{align}
    Z^n = \int d\mu(s^a) \exp \slr{\frac{\beta}{2N} \sum_a \lr{\sum_{i}\eta^1_i s^a_i}^2 + \frac{\beta}{2N} \sum_{\mu>1,a} \lr{\sum_{i}\eta^\mu_i s^a_i}^2 - \frac{\beta}{2N} (\lambda -1) \sum_a \lr{\sum_{i} s^a_i}^2}.
\end{align}
After inserting the order parameters via the integral representations of the delta functions \cite{Coolen} we end up with
\begin{align}
    Z^n = &\int d\mu(s^a) dm^a dr^a d^2 x_{\mu a} d^2 t_{\mu a} d^2 \overline{x}_{a} d^2 \overline{t}_{a}\nonumber\\
    &\exp (
    iN\sum_a r_a m_a - i \sum_a r_a \sum_i s^a_i + iN\sum_a \overline t_a \overline x_a - i \sum_a \overline t_a \sum_i \eta^1_i s^a_i +\nonumber\\
    &+ i \sum_{a,\mu>1} t_{\mu a} x_{\mu a} - \frac{i}{\sqrt N} \sum_{a,\mu>1} t_{\mu a} \sum_i \eta^\mu_i s^a_i + \frac{\beta N}{2} \sum_a \overline x^2_a + \frac{\beta}{2} \sum_{a,\mu>1} x^2_{\mu a} - \frac{\beta}{2} (\lambda-1) N \sum_a m_a^2)
\end{align}
The expectation over the quenched noise $\eta_{\mu>1}$ is easily computed and it reads
\begin{align}
    \expect{\eta_{\mu>1}} \exp\slr{-\frac{i}{\sqrt N} \sum_{a,\mu>1} t_{\mu a} \sum_i \eta^\mu_i s_i^a} = \exp\slr{-\frac{1}{2Nd} \sum_{i,\mu>1} \lr{\sum_a t_{\mu a} s^a_i}^2 + \mathcal{O}(N^{-2})}.
\end{align}
We therefore introduce the replica overlap order parameter $q_{ab}$ as
\begin{align}
    q_{ab} = \frac{1}{N}\sum_i s^a_i s^b_i.
\end{align}
The quenched replicated partition function then reads
\begin{align}
    \avg{Z^n} = &\int dm_a dr_a d^2 \overline{x}_{a} d^2 \overline{t}_{a} dq_{ab} dp_{ab} \nonumber\\
    &\exp (
    iN\sum_a r_a m_a + iN\sum_a \overline t_a \overline x_a + iN\sum_{ab} p_{ab} q_{ab} + \frac{\beta N}{2} \sum_a \overline x^2_a -  N\frac{\beta}{2} (\lambda-1) \sum_a m_a^2\nonumber\\
    &+K\ln \int  d^2 x_{\mu a} d^2 t_{\mu a} \exp\lr{i \sum_{a} t_{a} x_{a} - \frac{1}{2d}\sum_{ab} q_{ab} t_a t_b + \frac{\beta}{2} \sum_{a} x^2_{a} }\nonumber\\
    &+N\avg{\ln\lr{ \int d\mu(s^a)\exp\lr{- i \sum_a r_a s^a - i \sum_a \overline t_a \eta^1_i s^a - i \sum_{ab} p_{ab} s^a s^b} }}
    )
\end{align}
The Gaussian integrals in the third line of the latter equation can be computed and give
\begin{align}
    K\ln \int  d^2 x_{\mu a} d^2 t_{\mu a} \exp\lr{i \sum_{a} t_{a} x_{a} - \frac{1}{2d}\sum_{ab} q_{ab} t_a t_b + \frac{\beta}{2} \sum_{a} x^2_{a} } = -\frac{Kd}{2} \ln \det\lr{\delta_{ab}-\frac{\beta}{d} q_{ab}}.
\end{align}
In the thermodynamic limit, within the replica trick, the free energy of the Battaglia-Treves in the high storage regime of charts is given by the following extremal condition
\begin{align}
    \lim_{n\to0}\lim_{N\to\infty}\frac{\ln\avg{Z^n}}{n} = \lim_{n\to0}extr_{m,r,q,p,\overline t,\overline x} \frac{\Phi(m_a,r_a,q_{ab},p_{ab},\overline t_a,\overline x_a)}{n},
\end{align}
where 
\begin{align}
    \Phi =& \:iN\sum_a r_a m_a + iN\sum_a \overline t_a \overline x_a + iN\sum_{ab} p_{ab} q_{ab} + \frac{\beta N}{2} \sum_a \overline x^2_a -  N\frac{\beta}{2} (\lambda-1) \sum_a m_a^2+\nonumber\\
    &-\frac{Pd}{2} \ln \det\lr{\delta_{ab}-\frac{\beta}{d} q_{ab}} + N\avg{\ln\lr{ \int d\mu(s^a)\exp\lr{- i \sum_a r_a s^a - i \sum_a \overline t_a \eta^1_i s^a - i \sum_{ab} p_{ab} s^a s^b} }}
\end{align}
We proceed with the \emph{RS} ansatz \cite{MPV}:
\begin{align}
    &q_{ab} = q_1 \delta_{ab} + q_2 (1-\delta_{ab})\\
    &p_{ab} = \frac{i}{2} p_1 \delta_{ab} + \frac{i}{2} p_2 (1-\delta_{ab})\\
    &m_a = \overline m\\
    &\overline t_a = i \overline t\\
    &\overline x_a = i \overline x\\
    &r_a = i\frac{\beta}{2} r.
\end{align}
The \emph{RS} $\Phi$ functional then reads
\begin{align}
    \lim_{n\to0}\frac{\Phi^{RS}}{n} = -\frac{\beta}{2} r \overline m - \overline t \overline x + \frac{1}{2} (p_1 q_1 - p_2 q_2) + \frac{\beta}{2} \overline x^2 - \frac{\beta}{2} (\lambda-1) \overline m^2 -
    \lim_{n\to0}\frac{1}{n}\lr{\frac{\alpha d}{2}\ln\det\lr{\mathbf 1 - \frac{\beta}{d}\mathbf q}_{RS} -\: G_{RS}},
\end{align}
where $G_{RS}$ is
\begin{align}
    G_{RS} = \avg{\ln \int d\mu(s^a)\exp\lr{- i \sum_a r_a s^a - i \sum_a \overline t_a \eta s^a - i \sum_{ab} p_{ab} s^a s^b} }_\eta.
\end{align}
The latter reads
\begin{align}
    G_{RS} = n \avg{\int Dz \ln\lr{ 1 + \exp\lr{\frac{\beta}{2} r +\overline t \eta + \sqrt{p_2} z + \frac{1}{2}(p_1-p_2)} }}_\eta,
\end{align}
where we used $\int d\mu(s) = \sum_{s=0,1}$, as we restricted our investigation to Boolean spins. The determinant is
\begin{align}
    \ln\det\lr{\mathbf 1 - \frac{\beta}{d}\mathbf q}_{RS} = n \lr{\ln\lr{1-\frac{\beta}{d}(q_1-q_2)}-\frac{\frac{\beta}{d} q_2}{1-\frac{\beta}{d}(q_1-q_2)}} + \mathcal O(n^2).
\end{align}
Summing all the contributions we end up with
\begin{eqnarray}\nonumber
    A_{RS} &=& - \frac{\beta}{2} r \overline m - \overline t \overline x - \frac{1}{2} (p_1 q_1 - p_2 q_2) + \frac{\beta}{2} \overline x^2 - \frac{\beta}{2} (\lambda-1)\overline m^2 -
    \frac{\alpha d}{2} \ln\lr{1-\frac{\beta}{d}(q_1-q_2)}-\frac{\alpha \beta}{2} \frac{q_2}{1-\frac{\beta}{d}(q_1-q_2)}   \\ \label{Maruzzella}
    &+& \avg{\int Dz \ln\lr{ 1 + \exp\lr{\frac{\beta}{2} r +\overline t \eta + \sqrt{p_2} z + \frac{1}{2}(p_1-p_2)} }}_\eta.
\end{eqnarray}
To get rid off the auxiliary integration variables we start to extremize the free energy, namely we perform
\begin{align}
    &\pder[A^{RS}]{\overline m} = 0,\nonumber\\
    &\pder[A^{RS}]{\overline x} = 0,\nonumber
\end{align}
from which we obtain
\begin{align}
    &\frac{1}{2} r = (1-\lambda)\overline m\\
    &\overline t = \beta \overline x.
\end{align}
By inserting the above expressions in the free energy expression at the r.h.s. of eq. \eqref{Maruzzella}, the latter reads
\begin{align}
    \label{eq:ARS_ReplicaTrick}
    A_{RS} =& - \frac{\beta}{2} (1-\lambda)\overline m^2 - \frac{\beta}{2} \overline x^2 - \frac{1}{2} \lr{p_1 q_1 - p_2 q_2} - \frac{\alpha d}{2} \ln \lr{1-\frac{\beta}{d}(q_1-q_2)} + \frac{\alpha\beta}{2} \frac{q_2 }{1-\frac{\beta}{d}(q_1-q_2)} +\nonumber\\
    &+ \mathbb E_{\eta} \int Dz \ln \lr{1+ \exp\lr{ \beta (1-\lambda) \overline m + \beta \overline x \eta + \frac{1}{2}(p_1-p_2) + \sqrt{p_2} \:z }}.
\end{align}
Finally we eliminate also  the  auxiliary parameters $p_1$ and $p_2$ by differentiating $A_{RS}$ w.r.t. $q_1$ and $q_2$, \emph{i.e.} $\pder[A^{RS}]{q_1} = 0, \pder[A^{RS}]{q_2} = 0$, 
\begin{align}
    &\frac{1}{2}p_2 = \frac{\alpha \beta^2}{2d} \frac{q_2}{\lr{1-\frac{\beta}{d}(q_1-q_2)}^2},\\
    &\frac{1}{2}(p_1-p_2) = \frac{\frac{\alpha \beta}{2}}{1-\frac{\beta}{d}(q_1-q_2)}.
\end{align}
By inserting these expression in the r.h.s. of eq. \eqref{eq:ARS_ReplicaTrick} we end up with
\begin{align}
    A^{RS}(\alpha,\beta,\lambda) &= - \frac{\beta}{2} (1-\lambda)\overline m^2 - \frac{\beta}{2} \overline x^2 - \frac{\alpha \beta}{2} \frac{q_1-\frac{\beta}{d}(q_1-q_2)^2}{\lr{1-\frac{\beta}{d}(q_1-q_2)}^2} - \frac{\alpha d}{2}\ln\lr{1-\frac{\beta}{d}(q_1-q_2)}+\frac{\alpha\beta}{2} \frac{q_2}{1-\frac{\beta}{d}(q_1-q_2)} +\nonumber\\
    &+\avg{\int Dz \ln\lr{ 1 + \exp\lr{\beta(1-\lambda)\overline m +\beta \overline x \eta + \beta \frac{\frac{\alpha}{2} + \sqrt{\frac{\alpha q_2}{d}}\:z}{1-\frac{\beta}{d}(q_1-q_2)}}}}_\eta,
\end{align}
that is the same expression we obtained by the Guerra's route based on stochastic stability in the high storage limit (and that collapses to the low-storage expression obtained by the Guerras' route based on the mechanical analogy simply by forcing $\alpha$ to be zero): see Theorem $1$ in the main text.

\section{Free energy calculation via the Interpolation Method}\label{GuerraAPP}

Hereafter we prove the results reported in Propositions \ref{StreamingHSR} and \ref{CauchyHS} and in Theorem \ref{TzeTze}.
\newline
\newline
\begin{proof}(of Proposition  \ref{StreamingHSR}) 
Let us start with  a slightly more general definition of the interpolating Hamiltonian $\mathcal{H}(t)$ provided in Definition \ref{InterpoAcca}:
\begin{equation}
-\beta \mathcal{H}(t) = \frac{a_1(t)}{2} N \beta x_1^2 + \frac{a_2(t)}{2} N \beta (1-\lambda) m^2 + a_3(t)\sqrt{\frac{\beta}{N}}\sum_{i,\mu > 1}^{N,K} \eta_i^{\mu} s_i z_{\mu} + N\phi(t),     
\end{equation}
where we kept the specific interpolation coefficient as general functions $a_i(t), \ i \in (1,2,3),$ and let us perform the evaluation of the  $t$-derivative of $\mathcal A(t)$, this is computed as 
\begin{align}
    \mathcal A'(t) = \lim_{N\to \infty} N^{-1}\lr{ \frac{a_1'}{2} N\beta \avg{x_1^2} + \frac{a_2'}{2} N\beta \lr{1-\lambda} \avg{m^2} + a_3'\sqrt{ \frac{\beta}{N}} \sum_{i,\mu>1} \avg{\eta^\mu_i s_i x_\mu} + N \avg{\phi'}},
\end{align}
where we used the superscript $'$ to indicate the derivative with respect to $t$ in order to lighten the notation.\\
In order to evaluate $\avg{\eta^\mu_i s_i x_\mu}$ we exchange the expectation over the maps $\mathbb E_\eta$ with a $d$-variate Gaussian measure with the same mean and variance: $\mathcal N(0, \mathbf 1_d/d)$. This is done by noticing that $\avg{\eta^\mu_i s_i x_\mu}$ can be written as
\begin{align*}
    \sum_{i,\mu>1} \avg{\eta^\mu_i s_i x_\mu} = \sum_{i,\mu>1} \mathbb E_\eta \left[\eta^\mu_i \omega \lr{s_i x_\mu}\right] \sim \sum_{i,\mu>1}\mathbb E_z\lr{z^\mu_i s_i x_\mu},
\end{align*}
with $z^\mu_i\sim \mathcal N(0, \mathbf 1_d/d)$ and noticing that the introduced error, computed as $|\mathbb E_{\eta}\lr{\eta^\mu_i s_i x_\mu}-\mathbb E_z\lr{z^\mu_i s_i x_\mu}|$, vanishes in the thermodynamic limit by virtue of the Stein lemma \cite{Barra-JSP2008}. The variance of order $1/d$ is required by the condition $\mathbb E_\eta (\eta^2) = 1$. \\
Hence we are able to compute the $\avg{\eta^\mu_i s_i x_\mu}$ term by using the Gaussian integration by parts (namely by exploiting the Wick-Isserlis theorem \cite{GuerraNN}):
\begin{align}
    \sum_{i,\mu>1} \avg{\eta^\mu_i s_i x_\mu} 
    = \frac{1}{d} \sum_{i,\mu>1} \mathbb E_\eta \: \sum_{q=1}^d \pder[\omega \lr{s_i x_\mu^{(q)}}]{\eta^{\mu,(q)}_i} = \frac{b}{d} \:\sqrt{ \frac{\beta}{N}} NK \lr{\avg {q_{11} p_{11}} - \avg{q_{12} p_{12}}},
\end{align}
where, in the first equality, we use the Wick-Isserlis theorem component-wise, with components $x_\mu=\lr{x_\mu^{(q)}}_{q=1,..,d}$ and $\eta^\mu_i=\lr{\eta_i^{\mu,(q)}}_{q=1,..,d}$ .
This quantity is computed by observing that the partial derivative of the Boltzmann average $\omega (y)$ w.r.t. $\eta$ produces the difference of two averages which, apart from global multiplying factors, corresponds to the average of the squared argument, $\omega(y^2)$, and the square of the Boltzmann average, $\omega^2(y)$. 
In our case we use this relation component-wise to get
\begin{align}
    \sum_{q=1}^d \pder[\omega \lr{s_i x_\mu^{(q)}}]{\eta^{\mu,(q)}_i} = b\sqrt{ \frac{\beta}{N}} \sum_{q=1}^d \left[ \omega \lr{(x^{(q)}_\mu s_i)^2} - \omega^2 \lr{x^{(q)}_\mu s_i} \right].
\end{align}
The first term reads
\begin{align}
    \sum_{q=1}^d \omega \lr{(x^{(q)}_\mu s_i)^2} = \omega \lr{\sum_{q=1}^d  \lr{x^{(q)}_\mu}^2 \lr{s_i}^2} = \omega \lr{\lr{x_\mu}^2 \:\lr{s_i}^2};
\end{align}
in the last equality we used $\lr{x_\mu}^2 = \sum_{q=1}^d \lr{x^{(q)}_\mu}^2$, that is the definition of the Euclidean norm of $x_\mu$. This term can be rewritten as the scalar product of the first field replica  $x_\mu^1$ with itself times the square of the first spin replica $s_i^1$:
\begin{align}
    \omega \lr{\lr{x_\mu}^2\: \lr{s_i}^2} = \omega \lr{x_\mu^1 \cdot x_\mu^1\: s_i^1 s_i^1}.
\end{align}
In order to deal with the second term, we introduce the distinct spin replicas $s_i^1,s_i^2$ and the field replicas $x_\mu^{(q),1}, x_\mu^{(q),2}$ component-wise, namely:
\begin{align}
    \sum_{q=1}^d \omega^2 \lr{x^{(q)}_\mu s_i} = \sum_{q=1}^d \omega\lr{s_i^1 s_i^2\: x_\mu^{(q),1} x_\mu^{(q),2}},
\end{align}
which thus reads
\begin{align}
    \sum_{q=1}^d \omega\lr{s_i^1 s_i^2\: x_\mu^{(q),1} x_\mu^{(q),2}} =  \omega\lr{s_i^1 s_i^2 \sum_{q=1}^d x_\mu^{(q),1} x_\mu^{(q),2}} = \omega\lr{s_i^1 s_i^2 \: x_\mu^1 \cdot x_\mu^2}.
\end{align}
Finally, collecting all terms together we are able to write
\begin{align}
    \sum_{i,\mu>1} \mathbb E_\eta \: \sum_{q=1}^d \pder[\omega \lr{s_i x_\mu^{(q)}}]{\eta^{\mu,(q)}_i} = 
    \sum_{i,\mu>1} \mathbb E_\eta \lr{\omega \lr{s^1_is^1_i\: x^1_\mu \cdot x^1_\mu} - \omega \lr{s^1_i s^2_i \: x^1_\mu \cdot x^2_\mu}}.
\end{align}
The latter allows us to introduce the replica overlaps order parameters $Nq_{11}=\sum_i s^1_i s^1_i$, $Nq_{12}=\sum_i s^1_i s^2_i$, $K p_{11}=\sum_\mu x^1_\mu \cdot x^1_\mu$, $Kp_{12}=\sum_\mu x^1_\mu \cdot x^2_\mu$, to get
\begin{align}
    &\sum_{i,\mu>1} \mathbb E_\eta \omega \lr{s^1_is^1_i \:x^1_\mu \cdot x^1_\mu} = NK\avg {q_{11} p_{11}},\\
    &\sum_{i,\mu>1} \mathbb E_\eta \:\omega \lr{s^1_i s^2_i \: x^1_\mu \cdot x^2_\mu} = NK\avg{q_{12} p_{12}}.
\end{align}
The derivative of $\mathcal A(t)$, in their terms, reads as
\begin{align}
    \mathcal A' = \frac{a_1'}{2}\beta \avg{x_1^2} + \frac{a_2'}{2} \beta \lr{1-\lambda} \avg{m^2} + \frac{\beta \alpha}{2d} \:a_3 a_3' \lr{\avg {q_{11} p_{11}} - \avg {q_{12} p_{12}}} + \avg{\phi'},
\end{align}
where the load of the model $\alpha=\lim_{N\to \infty} K/N$ has been introduced.The latter equation contains only two-points correlation functions, \emph{e.g.} $\avg{x_1^2},\avg{q_{12}},..$, due to the presence of second order interactions in the Hamiltonian: these terms are generally hard to compute and the general strategy of the Guerra's interpolation technique is to balance them out  with appropriate counter-terms that can be obtained by differentiating the $\avg{\phi(t)}$ functional. For example, the two-point function $\avg{x_1^2}$ can be produced by a term of the form $\sum_{\mu>1} x_\mu^2$, while a term like $\sum_i h_i s_i$ (with $h_i \in \mathcal N(0,1)$) produces a contribution of the form of $\avg{q_{11}} - \avg{q_{12}}$ (apart from global factors). This motivates our choice of the $\phi(t)$ functional, which is thus expressed as a sum of one-body and two-bodies terms as follows: 
\begin{align}
    N\phi(t) = \phi_1(t)\sum_{\mu>1} x_\mu^2 + \phi_2(t) \sum_\mu \rho_\mu x_\mu + \phi_3(t) \sum_i h_i s_i + \phi_4(t) \sum_i \eta^1_i s_i + \phi_5(t) \sum_i s_i^2 + \phi_6(t) \sum_i s_i,
\end{align}
where the $t-$dependence is delegated to the parameters $\phi_1(t),\phi_2(t),\phi_3(t),\phi_4(t),\phi_5(t),\phi_6(t)$, which have to respect the boundary conditions $\phi_1(1)=\phi_2(1)=\phi_3(1)=\phi_4(1)=\phi_5(1)=\phi_6(1)=0$. The derivative of $\phi(t)$ then reads
\begin{align}
    \phi' = \phi_1'\alpha\avg{p_{11}} + \phi_2 \phi_2' \alpha\lr{\avg{p_{11}} - \avg{p_{12}}} + \phi_3 \phi_3' \lr{\avg{q_{11}} - \avg{q_{12}}} + \phi_4' \avg{x_1} + \phi_5' \avg{q_{11}} + \phi_6' \avg{m}.
\end{align}
Collecting all the terms, the derivative of $\mathcal A(t)$ finally reads
\begin{align}\label{eq:Ader}
    \mathcal A' = \frac{a_1'}{2} \beta \avg{x_1^2} + &\frac{a_2'}{2} \beta \lr{1-\lambda}\avg{m^2} + \alpha\beta \frac{a_3 a_3'}{d} \lr{\avg {q_{11} p_{11}} - \avg {q_{12} p_{12}}} + \nonumber\\
    &\phi_1'\alpha\avg{p_{11}} + \phi_2 \phi_2' \alpha\lr{\avg{p_{11}} - \avg{p_{12}}} + \phi_3 \phi_3' \lr{\avg{q_{11}} - \avg{q_{12}}} + \phi_4' \avg{x_1} + \phi_5' \avg{q_{11}} + \phi_6' \avg{m}.
\end{align}
Next, we treat the order parameters expectations in terms of their fluctuations, under the RS hypothesis (remember Definition \ref{selfaverage}) these fluctuations are vanishing in the thermodynamic limit $N\to\infty$, namely
\begin{align}
    &\avg{\Delta^2 x} = \avg{(x_1-\overline x)^2}\to 0,\\
    &\avg{\Delta^2 m} = \avg{(m-\overline m)^2}\to 0,\\
    &\avg{\Delta p_1 \Delta q_1}= \avg{(q_{11}-\overline q_1)(p_{11}-\overline p_1)}\to 0,\\
    &\avg{\Delta p_2 \Delta \overline q_2} = \avg{(q_{12}-\overline q_2)(p_{12}-\overline p_2)}\to 0,
\end{align}
where $\overline x =\avg{x}$, $\overline m = \avg{m}$ and $\overline p_1 = \avg{p_{11}}$, $\overline p_2 = \avg{p_{12}}$, $\overline q_1 = \avg{q_{11}}$, $\overline q_2 = \avg{q_{12}}$.\\
Collecting the homogeneous terms in equation \eqref{eq:Ader} and systematically eliminating the terms containing the expectation values of the order parameters \eqref{eq:one_order}-\eqref{eq:three_order}, we obtain the following system of coupled differential equations to be solve in the interval $t\in(0,1)$:
\begin{align}
    &\phi_6' + \beta (1-\lambda) a_2' \overline m = 0,\\
    &\phi_4' + \beta a_1' \overline x = 0,\\
    &\phi_1' + \beta \frac{a_3 a_3'}{d} \overline q_1  + \phi_2 \phi_2' = 0,\\
    &\phi_5' + \alpha \beta \frac{a_3 a_3'}{d} \overline p_1 + \phi_3 \phi_3' = 0,\\
    &\phi_2 \phi_2' + \beta \frac{a_3 a_3'}{d} \overline q_2 = 0,\\
    &\phi_3 \phi_3' + \alpha \beta \frac{a_3 a_3'}{d} \overline p_2 = 0.
\end{align}
Without loosing generality we can also impose the conditions
\begin{align}
    &a_1' = 1,\\
    &a_2' = 1,\\
    &a_3 a_3' = 1/2.
\end{align}
The solutions reads 
\vspace{2mm}\\
\begin{minipage}{.4\textwidth}
    \begin{subequations}
    \begin{align*}[left = {\empheqlbrace}]
    &\phi_1(t) = \frac{\beta}{2d}(\overline q_1-\overline q_2)(1-t),\\
    &\phi_2(t) = \sqrt{\frac{\beta}{d} \overline q_2\:(1-t)},\\
    &\phi_3(t) = \sqrt{\frac{\alpha \beta}{d} \overline p_2\: (1-t)},\\
    &\phi_4(t) = \beta \overline x\: (1-t),\\
    &\phi_5(t) = \frac{\alpha\beta}{2d}(\overline p_1-\overline p_2)(1-t),\\
    &\phi_6(t) = \beta (1-\lambda) \overline m\: (1-t).
    \end{align*}
    \end{subequations}
\end{minipage}
\begin{minipage}{.4\textwidth}
    \begin{subequations}
    \begin{align*}[left = {\empheqlbrace}]
    &a_1(t) = t,\\
    &a_2(t) = t,\\
    &a_3(t) = \sqrt{t},
    \end{align*}
    \end{subequations}
\end{minipage}\\\vspace{2mm}\\
The boundary conditions are finaly derived from the interpolation conditions $\phi(1)=0$ and $a_1(0) = a_2(0) = a_3(0) = 0$. 
\newline
Once all the auxiliary functions have been made fully explicit, we can sum up all the contributions to the $t-$derivative of $\mathcal A$ that then reads
\begin{align}
    \mathcal A' = \frac{\beta}{2} (1-\lambda)\lr{\avg{\Delta^2 m} - \overline m^2} + \frac{\beta}{2} \lr{\avg{\Delta^2 x} - \overline x^2} + &\frac{\alpha \beta}{2d} \lr{\avg{\Delta p_1 \Delta q_1} - \overline p_1 \overline q_1} - \frac{\alpha \beta}{2d}\lr{\avg{\Delta p_2 \Delta q_2} - \overline p_2 \overline q_2}. 
\end{align}
Under the RS assumption, the fluctuations of the order parameters vanish and the $t-$derivative of $\mathcal A$ becomes the expression reported in Proposition \ref{StreamingHSR}.

\end{proof}
\begin{proof}(of Proposition \ref{CauchyHS})
We are left with the Cauchy condition to calculate: the interpolating free energy has to be evaluated at $t=0$ but, before doing that, we write explicitly the partition function at $t=0$ for the sake of clearness: 
\begin{align}\nonumber
    \mathcal Z(t=0) = \int d\mu(s) Dx \:\exp\biggl(\phi_1(0) \sum_{\mu>1} x_\mu^2 + \phi_2(0) \sum_{\mu>1}\rho_\mu x_\mu &+\phi_3(0) \sum_i h_i s_i + \phi_4(0) \sum_i \eta_i^1 s_i \\
    &+\phi_5(0) \sum_i s_i^2 + \phi_6(0) \sum_i s_i\biggl).
\end{align}
After performing the Gaussian integrals in $D^d x = \prod_{\mu>1} D^d x_\mu$, where $Dx_\mu = \frac{d^d x_\mu}{\sqrt{2\pi}} e^{-\frac{x_\mu^2}{2}}$ is the Gaussian measure, and using Boolean neurons $s_i = \{0,1\}$, we end up with:
\begin{align}
    \mathcal Z(t=0) = \slr{ \prod_{\mu>1} \lr{\frac{2\pi}{1-2\phi_1(0)}}^{d/2} \exp\lr{\frac{1}{2} \frac{\phi^2_2(0) \rho_\mu^2}{1-2\phi_1(0)}} }\times \prod_i\slr{ 1+ \exp\lr{\phi_3(0) h_i + \phi_4(0) \eta_i^1 + \phi_5(0) + \phi_6(0)}}.
\end{align}
We can finally rearrange all the terms as provided by Proposition \ref{CauchyHS}.\end{proof}



\printbibliography

\end{document}